\documentclass{article}
\usepackage{praneeth}
\usepackage[ruled, linesnumbered, vlined, norelsize]{algorithm2e}
\usepackage[margin=1in]{geometry}
\usepackage{mlmodern}

\title{Faster Algorithms for Schatten-\texorpdfstring{$p$}{p} Low Rank Approximation} 

\author{Praneeth Kacham \\ CMU and Google Research \and David P. Woodruff \\ CMU}

\newcommand{\machine}{\textnormal{mach}}

\newcommand{\lp}[1]{\|#1\|_{S_p}}
\newcommand{\gap}{\textnormal{\textsf{gap}}}
\renewcommand{\tilde}{\widetilde}
\newcommand{\fl}{\textnormal{fl}}
\newcommand{\AppxPCA}{\texttt{AppxPCA}}
\newcommand{\Proj}{\textnormal{Proj}}
\newcommand{\pca}{\textnormal{pca}}
\newcommand{\colspace}{\textnormal{colspace}}
\begin{document}
\maketitle

\begin{abstract}
    We study algorithms for the Schatten-$p$ Low Rank Approximation (LRA) problem. First, we show that by using fast rectangular matrix multiplication algorithms and different block sizes, we can improve the running time of the algorithms in the recent work of Bakshi, Clarkson and Woodruff (STOC 2022).    
    We then show that by carefully combining our new algorithm with the algorithm of Li and Woodruff (ICML 2020), we can obtain even faster algorithms for Schatten-$p$ LRA. 
    
        While the block-based algorithms are fast in the real number model, we do not have a stability analysis which shows that the algorithms work when implemented on a machine with polylogarithmic bits of precision. We show that the LazySVD algorithm of Allen-Zhu and Li (NeurIPS 2016) can be implemented on a floating point machine with only logarithmic, in the input parameters, bits of precision. As far as we are aware, this is the first stability analysis of any algorithm using $O((k/\sqrt{\varepsilon})\poly(\log n))$ matrix-vector products with the matrix $A$ to output a $1+\varepsilon$ approximate solution for the rank-$k$ Schatten-$p$ LRA problem. 
\end{abstract}

\section{Introduction}
Low Rank Approximation (LRA) is an important primitive in large scale data analysis. Given an $m \times n$ matrix $A$, and a rank parameter $k$, the task is to find a rank-$k$ matrix $B$ that minimizes $\|A - B\|$ where $\| \cdot \|$ is some matrix norm. Typically, we also require that the algorithms output a factorization $B = XY$ such that $X \in \R^{m \times k}$ and $Y \in \R^{k \times n}$. Such a factorization lets us compute the product $Bz$ with an arbitrary vector $z$ in time $O(k(n+m))$ which can be significantly smaller than the $\nnz(A)$ time required to multiply a vector with the original matrix $A$. Here $\nnz(A)$ denotes the number of non-zero entries of the matrix $A$. Thus, replacing $A$ with a low rank approximation can make downstream tasks much faster. Additionally, if the matrix $A$ has a low rank structure but is corrupted by noise, a low rank approximation of $A$ can recover the underlying structure under suitable assumptions on the noise. We note that many low rank approximation algorithms, including ours, compute a rank-$k$ orthonormal matrix $W$ such that $\|A(I-W\T{W})\|$ is small and then define $X = AW$ and $Y = \T{W}$. 

In this paper, the error metric we consider is given by the Schatten-$p$ norm for $p \ge 1$. Given a matrix $M$, the Schatten-$p$ norm of $M$ denoted by $\lp{M}$ is defined as $(\sum_{i} \sigma_i(M)^p)^{1/p}$ where $\sigma_i(M)$ denotes the $i$-th singular value of $M$.  Note that Schatten-$2$ norm is the same as the Frobenius norm, denoted by $\frnorm{M} = (\sum_{i,j}M_{ij}^2)^{1/2}$ and the Schatten-$\infty$ norm is the same as the operator norm, denoted by $\opnorm{M} = \max_{x \ne 0}\opnorm{Mx}/\opnorm{x}$. In the presence of outliers, the Schatten-$1$ norm, $\sum_i\sigma_i(M)$, is considered to be more robust since the errors introduced by the outliers are not ``squared'' as it is done in the case of the Frobenius norm.

The Schatten-$p$ norm low rank approximation problem asks to find a rank-$k$ matrix $B$ that minimizes $\lp{A - B}$. As the Schatten-$p$ norms are unitarily invariant, we have from Eckart-Young-Mirsky's theorem that $\lp{A - A_k} = \min_{\text{rank-}k\, B}\lp{A - B}$ for all $p \ge 1$, where $A_k$ is the matrix obtained by truncating the Singular Value Decomposition (SVD) of $A$ to only the top $k$ singular values. This implies that a single matrix $A_k$ is a \emph{best} rank-$k$ approximation for $A$ for all values of $p$. However, computing the SVD of an $m \times n$ matrix takes $O(\min(mn^{\omega-1}, nm^{\omega-1}))$ time (see Appendix~\ref{sec:svd-complexity}), where $\omega$ is the matrix multiplication exponent. This time complexity is prohibitive when $m$ and $n$ are large. Thus, we relax the requirements and ask for a rank-$k$ matrix $B$ satisfying $\lp{A-B} \leq (1+\varepsilon)\lp{A - A_k}$ in the hope of obtaining faster algorithms than the SVD.

While a single matrix $A_k$ is a best low rank approximation for $A$ in all Schatten-$p$ norms, it is not the case for approximate solutions, i.e., if $B$ is a rank-$k$ matrix that satisfies $\lp{A - B} \le (1+\varepsilon)\lp{A - A_k}$ for some $p$, it may not be the case that $\|A-B\|_{S_q} \le (1+\varepsilon)\|A-A_k\|_{S_q}$ for $q \ne p$. Thus, many approximation algorithms for Schatten-$p$ LRA are tailored to the particular value $p$. There are two different lines of works for Schatten-$p$ LRA in the literature: (i) Sketching based algorithms of Li and Woodruff \cite{li2020input} and (ii) Iterative algorithms of Bakshi, Clarkson and Woodruff \cite{bakshi2022low}. We summarize the running times of the algorithms in Table~\ref{tab:results-summary}.
\begin{table}[b]
    \centering
    \begin{tabular}{l l}\toprule
    & Time Complexity\\ \midrule
    Li and Woodruff \cite{li2020input} ($p \in [1,2)$) & 
    \begin{tabular}{@{}c@{}}$O(\nnz(A)\log n) + \tilde{O}_p(mk^{2(\omega-1)/p}/\varepsilon^{(4/p-1)(\omega-1)}$) \\ $+ \tilde{O}_p(k^{2\omega/p}/\varepsilon^{(4/p-1)(2\omega+2)})$\end{tabular}
    \\
    Li and Woodruff \cite{li2020input} ($p > 2$) & 
    \begin{tabular}{@{}c@{}}
    $O(\nnz(A)\log n) + \tilde{O}_p(n^{\omega(1-2/p)}k^{2\omega/p}/\varepsilon^{2\omega/(p+2)})$ \\ $ + \tilde{O}_p(mn^{(\omega-1)(1-2/p)}(k/\varepsilon)^{2(\omega-1)/p})$\end{tabular}\\
    Bakshi et al. \cite{bakshi2022low} & $O(p^{1/6}\varepsilon^{-1/3}\nnz(A)k\log(n/\varepsilon) + mp^{(\omega-1)/6}k^{\omega-1}\varepsilon^{-(\omega-1)/3})$\\\bottomrule
    \end{tabular}
    \caption{Running times for $1+\varepsilon$ rank-$k$ Schatten-$p$ LRA algorithms for $m \times n$ matrices assuming $m \ge n$.}
    \label{tab:results-summary}
\end{table}
The sketch-based algorithms are usually non-adaptive and the iterative algorithms adaptively pick their matrix-vector product queries depending on the results in the previous round which makes them powerful as we can see from the superior running time over sketch-based algorithms when we desire solutions with small $\varepsilon$.

\subparagraph*{Sketching Algorithms.} Li and Woodruff \cite{li2020input} gave (almost) input-sparsity time algorithms for Schatten-$p$ LRA,  extending the earlier input-sparsity time algorithms for Frobenius norm LRA from \cite{cw-lra}. For $p < 2$, their algorithm runs in $\tilde{O}(\nnz(A) + \max(m, n) \cdot \poly(k/\varepsilon))$ time and for $p > 2$, their algorithm runs in $\tilde{O}(\nnz(A) + \max(m, n) \cdot \min(m, n)^{\alpha_p}\poly(k/\varepsilon))$ time, where $\alpha_p = (\omega-1)(1-2/p)$. Note that for the current value of $\omega \approx 2.37$, their algorithm runs in $\Omega(mn)$ time for $p \ge 7.4$ and hence is not an ``input-sparsity time'' algorithm but for all constant $p,k,\varepsilon$, their algorithm runs in $o(\min(mn^{\omega-1}, nm^{\omega-1})$ time and therefore is faster than computing the SVD.

\subparagraph*{Iterative Algorithms.} Recently, Bakshi, Clarkson, and Woodruff \cite{bakshi2022low} gave an iterative algorithm for Schatten-$p$ LRA. Their algorithm runs the Block Krylov iteration algorithm of Musco and Musco \cite{musco2015randomized} at \emph{two} different block sizes for different number of iterations respectively. They show that the algorithm succeeds in computing a low rank approximation at one of the block sizes and show how to compute which block size succeeds in computing the approximation. For Schatten-$p$ LRA, their algorithm requires $O(kp^{1/6}\poly(\log n)/\varepsilon^{1/3})$ matrix-vector products with the matrix $A$ and hence can be implemented in $\tilde{O}(\nnz(A)kp^{1/6}/\varepsilon^{1/3})$ time. At a high level, their algorithm runs the Block Krylov iteration algorithm with block size $k$ for $O(p^{1/6}\varepsilon^{-1/3}\poly(\log n))$ iterations and with block size $O(p^{-1/3}\varepsilon^{-1/3}k)$ for $O(\sqrt{p}\poly(\log n))$ iterations. They set these parameters such that the algorithm requires an overall same number of matrix-vector products with $A$ at both block sizes. They argue that for a matrix with a ``flat'' spectrum, the low rank approximation computed by the block size $k$ algorithm is a $1+\varepsilon$ approximation and for a matrix with a ``non-flat'' spectrum, the solution computed by block size $O(p^{-1/3}\varepsilon^{-1/3}k)$ algorithm is a $1+\varepsilon$ approximation. 

\subparagraph*{Comparison.} As we can see from Table~\ref{tab:results-summary}, the running times of these algorithms depend in a quite complicated way on the parameters $\nnz(A)$, $m$, $n$, $\varepsilon$ and $p$. Throughout the paper, we assume that $m = n$, $\nnz(A) = n^2$ (i.e., the matrix $A$ is dense) and $k \le n^{c}$ for a small constant $c$ so that $k \ll n$. In some cases, where sparsity in the datasets cannot be well exploited, such as when processing the datasets using GPUs, it is natural to analyze the time complexities of the algorithms and compare the performances assuming that the inputs are dense.

For $p \in [1,2)$, we have that the time complexity of the algorithm of \cite{li2020input} is $O(n^2\log n + n\poly(k)/\varepsilon^{(4/p-1)(\omega-1)} + \poly(k)/\varepsilon^{(4/p-1)(2\omega+2)})$ and the time complexity of the algorithm of \cite{bakshi2022low} is $O(\varepsilon^{-1/3}n^2 k\log(n) + n\poly(k)/\varepsilon^{(\omega-1)/3})$. We see that only when
\begin{align*}
    1/\varepsilon > {n^{\frac{1}{(4/p-1)(\omega+1)-1/6}}},
\end{align*}
the algorithm of \cite{bakshi2022low} is faster than the sketching based algorithm of \cite{li2020input}.
For $\omega \approx 2.371$ and $p=1$, the above is achieved only when $1/\varepsilon \ge n^{\approx 0.1}$. Hence, in the  high accuracy regime, the algorithm of \cite{bakshi2022low} is faster than that of the sketching based algorithm of \cite{li2020input}. For other values of $p \in [1,2)$, $\varepsilon$ has to be even smaller than $1/n^{0.1}$ for the algorithm of \cite{bakshi2022low} to be faster than the algorithm of \cite{li2020input}.

For comparing the algorithms in the case $p > 2$, first we pick $\varepsilon$ to be a constant and obtain that the running time of the algorithm of \cite{li2020input} is $O(n^2 \log n + n^{1+(\omega-1)(1-2/p)}\poly(k))$ and the algorithm of \cite{bakshi2022low} has a running time of $O(p^{1/6}n^2k \log(n))$. Thus, as long as $(\omega-1)(1-2/p) \le 1$, the sketch-based algorithm is faster than the iterative algorithm. We call $p$ such that $(\omega-1)(1-2/p) \le 1$, the \emph{crossover point} from ``sketch'' to ``iterative''. For the current value of $\omega \approx 2.371$, the crossover point is $\approx 7.39$.

Now consider the case of $\varepsilon = 1/n$ and constant $p$. The iterative algorithm of \cite{bakshi2022low} has a running time of $O(n^{2+1/3}k\log (n))$ and the sketch based algorithm of \cite{li2020input} has a running time of $O(n^{\omega}\poly(k))$ and thus offers no improvement over the na\"ive SVD algorithm. This again shows that in the high precision regime, the small dependence on $\varepsilon$ in the running time of the algorithm of \cite{bakshi2022low} is crucial to obtain better than $O(n^{\omega})$ time algorithm. Overall, we summarize the comparison between the algorithms in Table~\ref{tab:results-comparison}.
    \begin{table}
    \centering
    \begin{tabular}{l c c}\toprule
         & Small $\varepsilon$ ($\approx 1/n$) & Large $\varepsilon$\\ \midrule
         $p \in [1,2)$ & Iterative & Sketching  \\
         $2 < p < \text{crossover}$ & Iterative   & Sketching \\
         $p > \text{crossover}$ & Iterative  & Iterative \\ \bottomrule
    \end{tabular}
    \caption{In the case of $m = n$, $\nnz(A) = n^2$ and $k = n^{o(1)}$, the table lists which of the previous works is asymptotically faster for the current value of $\omega \approx 2.371$. \textbf{Iterative} algorithm refers to the algorithm of  \cite{bakshi2022low} and the \textbf{Sketching} algorithm refers to the algorithm of \cite{li2020input}. In the above, crossover $\approx 7.4$.}
    \label{tab:results-comparison}
\end{table}

\subparagraph*{Our Improvements.} We first \emph{improve} the time complexity of the iterative algorithm of \cite{bakshi2022low} for \emph{all} parameter regimes. While the focus of their paper was to minimize the number of matrix-vector products required, we observe that by using fast rectangular matrix multiplication algorithms, we can obtain even faster algorithms using their technique of running the block Krylov iteration algorithm at different block sizes. Fast rectangular matrix multiplication algorithms let us obtain a different block-size vs iteration trade-off giving us faster algorithms. This algorithm directly achieves the fastest running times for small $\varepsilon$ since we improve upon \cite{bakshi2022low} in all regimes.

We saw above that for constant $\varepsilon$, the sketch based algorithm takes only $O(n^2 \log n)$ time when $p \lesssim 7.4$ and hence cannot be improved upon over asymptotically by more than $\polylog(n)$ factors in that regime. We show that using a combination of our fast iterative algorithm and the algorithm of \cite{li2020input} gives an algorithm that runs in near-linear time\footnote{Note the near-linear here means $\tilde{O}(n^2)$ as the input-matrix is assumed to have $n^2$ nonzero entries.} for all $p \lesssim 22$ for appropriate $\varepsilon$ values extending the values of $p$ for which a Schatten-$p$ LRA can be computed in $O(n^2  \log n)$ time, when the rank parameter $k \le n^c$.

Our combined algorithm works as follows: to solve a sub-problem in the algorithm of \cite{li2020input}, we run our improved iterative algorithm for Schatten-$p$ LRA with accuracy parameter $\varepsilon = 1/n$. As our improved iterative algorithm has a better dependence on $\varepsilon$ than earlier algorithms, we obtain a faster algorithm for solving the sub-problem and hence obtain an $O(n^2 \log n)$ time algorithm for all $p \lesssim 22$. Thus, improving the performance of iterative algorithms in the small $\varepsilon$ regime let us obtain faster algorithms overall in the large $\varepsilon$ regime!
\paragraph*{Numerically Stable Algorithms}
While the algorithm of \cite{bakshi2022low} and our modification give fast algorithms for Schatten-$p$ Low Rank Approximation, it is not known if the Block Krylov iteration algorithm is stable when implemented on a floating point machine with $O(\log (n/\varepsilon))$ bits of precision. It is a major open question in numerical linear algebra to show if the Block Krylov iteration algorithm is stable. Obtaining fast algorithms that provably work on finite precision machines is a tricky problem in general. We note that until the recent work of Banks, Garza-Vargas, Kulkarni and Srivastava \cite{banks2022pseudospectral}, it was not clear if an eigendecomposition of a matrix could be computed in $\tilde{O}(n^{\omega})$ time on a finite precision machine. Building on these ideas, another recent work \cite{sobczyk2023hermitian} obtains fast and stable algorithms for the generalized eigenvalue problem.  The sketch-and-solve methods, such as the algorithm of \cite{li2020input}, are usually stable as the operations do not blow up the magnitude of the entries. As we note above, for large $p$, the algorithms in \cite{li2020input} are not input-sparsity time and hence an important question is if there are any stable input-sparsity time algorithms for large $p$. We answer this question in affirmative by showing that the LazySVD algorithm of \cite{allen2016lazysvd} can be stably implemented on a floating point machine with $O(\log m\kappa/\varepsilon)$ bits of precision where $\kappa = \sigma_1(A)/\sigma_{k+1}(A)$. The LazySVD algorithm computes a low rank approximation for all $p \ge 2$. 

Similar to the Block Krylov iteration algorithm, LazySVD also needs $O(k\poly(\log n)/\sqrt{\varepsilon})$ matrix-vector products with $A$. Additionally, the factorization output by LazySVD is \emph{simultaneously} a $1+\varepsilon$ approximation for all $p \ge 2$. To find a rank-$k$ approximation of $A$, the LazySVD algorithm first computes a unit vector $v$ which is an approximation to the top eigenvector of $\T{A}A$. Then the algorithm deflates $\T{A}A$ and forms the matrix $(I-v\T{v})\T{A}A(I-v\T{v})$ and proceeds to find an approximation to the top eigenvector of $(I-v\T{v})\T{A}A(I-v\T{v})$ and so on for a total of $k$ rounds. The authors show that the span of $k$ vectors found across all the iterations contains a $1+\varepsilon$ approximation if the eigenvector approximations satisfy an appropriate condition. Thus, to implement the LazySVD algorithm on a floating point machine, we first need a stable routine that can compute approximations to the top eigenvector of a given matrix. We show that such a routine can be implemented stably using the Lanczos algorithm \cite{stability-of-lanczos}. We additionally modify the LazySVD algorithm and show that the modification allows us to compute matrix-vector products with the deflated matrix to a good enough approximation which lets the Lanczos algorithm compute an approximation to the top eigenvector of the deflated matrix. Our slight modification to LazySVD turns out to be important in making the stability analysis go through. 

The novelty of our stability analysis is that instead of showing each of the vectors $\tilde{v}_1,\ldots, \tilde{v}_k$ computed by a finite precision algorithm are close to the vectors $v_1, \ldots, v_k$ that would be computed by an algorithm with unbounded precision, we essentially argue that for all $i$, the projection matrices onto the subspaces spanned by $\tilde{v}_1,\ldots, \tilde{v}_i$ and $v_1, \ldots, v_i$ are close using induction. This change makes the stability analysis work with only a polylogarithmic number of bits of precision whereas showing all $\tilde{v}_i$s are individually close to corresponding $v_i$s would require polynomially many bits of precision.

\subsection{Our Results}
In the following, $\alpha$ denotes the constant such that an arbitrary $n \times n$ matrix can be multiplied with an arbitrary $n \times n^{\alpha}$ matrix using $O(n^{2+\eta})$ arithmetic operations for any constant $\eta > 0$. The matrix multiplication exponent $\omega$ is the smallest constant such that an arbitrary $n \times n$ matrix can be multiplied with an arbitrary $n \times n$ matrix using $O(n^{\omega+\eta})$ arithmetic operations for any constant $\eta > 0$. For simplicity, we ignore the constant $\eta$, and write as if the matrices can be multiplied in $O(n^{\omega})$ time. We define $\beta \coloneqq (\omega-2)/(1-\alpha)$. Note that $\beta \le 1$.\footnote{See Section~\ref{subsec:fast-multi}.}
\begin{theorem}[Informal, Theorem~\ref{thm:final-theorem-our-algorithm}]
	Given an $n\times n$ matrix $A$, a rank parameter $k$ and an accuracy parameter $\varepsilon$, there is an algorithm that outputs a rank-$k$ orthonormal matrix $W$ that with probability $\ge 0.9$ satisfies,
$
		\lp{A(I-W\T{W})} \le (1+O(\varepsilon))\lp{A - A_k}.
$
If $k \le \varepsilon \cdot n^{\alpha}$, then the algorithm runs in $\tilde{O}(\sqrt{p}n^{2+\eta})$ time for any constant $\eta > 0$.
\label{thm:informal-main-theorem}
\end{theorem}

Combining the algorithm in the above theorem and the algorithm of \cite{li2020input}, we obtain the following result:
\begin{theorem}[Informal, Theorem~\ref{thm:combination}]
	Given an $n \times n$ matrix $A$, a rank parameter $k$ independent of $n$ and any constant $\eta > 0$, there is a randomized algorithm that runs in time $\tilde{O}((n^{1-2/p})^{2+\eta+(1-\alpha)\beta/(1+2\beta)}\poly(1/\varepsilon
) + n^2)$ and outputs a rank-$k$ projection $\hat{Q}$ that satisfies
$
		\|A(I-\hat{Q})\|_{S_p}^p \le (1+\varepsilon)\|A-A_k\|_{S_p}^p,
$
with probability $\ge 0.9$
\end{theorem}
The above theorem shows that for all $p$ at most a suitable constant, the algorithm runs in $\tilde{O}(n^2)$ time for $\varepsilon > 1/n^{c_p}$ for a small enough constant $c_p$ and hence is faster than using the algorithm of \cite{li2020input} or the algorithm in Theorem~\ref{thm:informal-main-theorem}. 

The following result shows that our modification of LazySVD can be stably implemented on a floating point machine.
\begin{theorem}[Informal, Theorem~\ref{thm:lazysvd-stability}]
    Given an $n \times d$ matrix $A$ with condition number $\kappa(A) = \sigma_1(A)/\sigma_{k+1}(A)$, an accuracy parameter $\varepsilon$, a rank parameter $k$ and probability parameter $\eta$, if the machine precision $\varepsilon_{\machine} \le \poly(\varepsilon\eta/n\kappa(A))$, then there is an algorithm that outputs a $d \times k$ matrix $V_k$ such that $\kappa(V_k) \le 4$ and with probability $\ge 1 - \eta$, for all $p \in [2, \infty]$,
\[\|A(I - \Proj_{\colspace(V_k)})\|_{S_p} \le (1 + O(\varepsilon))\|A - A_k\|_{S_p},\]
    and runs in time $O(\frac{\nnz(A)k}{\sqrt{\varepsilon}}\poly(\log(d\kappa(A)/\varepsilon\eta)) + d\poly(k, \log(dk/\eta\varepsilon)))$.
\end{theorem}
In the above theorem, $\Proj_{\text{colspace}(M)}$ denotes the orthogonal projection matrix onto the column space of $M$.

\subsection{Implications to Practice}
While the theoretical fast rectangular matrix  multiplication algorithms are not practically efficient, the message of this paper is that by optimizing for the number of matrix-vector products as in \cite{bakshi2022low}, we are leaving a lot of performance on the table. In modern computing architectures, multiplying an $n \times n$ and an $n \times b$ matrix is, for example, much faster than $b$ times the time required to multiply the $n \times n$ and an $n \times 1$ vector because of data locality and the opportunities for parallelization. Thus, in the algorithm of \cite{bakshi2022low}, running the block size $k$ version for fewer iterations while increasing the larger block size $b$ can give faster algorithms in practice than using the parameters that optimize for the number of matrix-vector products. We include a small experiment in the appendix which compares the time required to compute the product of an $n \times n$ matrix with matrices that have different numbers of columns.

LazySVD with our stability analysis uses a similar number of matrix vector products as the widely used Block Krylov iteration algorithm while requiring only polylogarithmic bits of precision. While as mentioned above, block-based algorithms such as Block Krylov iteration can be much faster than single-vector algorithms such as LazySVD and our modification of it, it is only the case when the matrix is directly given to us. When the matrix is implicitly defined in other ways (for e.g., as the Hessian of a neural network where we can efficiently compute Jacobian-Vector products), the difference in performance between block-based algorithms and single-vector algorithms is less pronounced. When guarantees of stability are required, the fastest algorithms in practice for Low Rank Approximation should use some combination of sketching as in \cite{li2020input} to reduce dimension stably and then use our modification of LazySVD algorithm to find the necessary top $k$ subspace. 
\section{Preliminaries}
\subsection{Notation} For a positive integer $n$, we use $[n]$ to denote the set $\set{1,\ldots,n}$. We use the notation $\tilde{O}(f(n))$ to denote $O(f(n)\poly(\log(f(n))))$ and $\tilde{O}_q(f(n))$ to hide the multiplicative factors that depend only on the parameter $q$. For a vector $x$, we use $\|{x}\|_2 = (\sum_i |x_i|^2)^{1/2}$ to denote the Euclidean norm of $x$.
Given an $m \times n$ matrix $A$, we use $A_{i,j}$ to denote the entry in the index $(i,j)$ of $A$. We use $A_{i*}$ to denote the $i$-th row of $A$ and $A_{*j}$ to denote the $j$-th column. We identify the multiplication of an $m \times n$ matrix with an $n \times k$ matrix with the notation $[m, n, k]$. For a matrix $A$, we use $\colspace(A)$ to denote the vector space $\setbuilder{Ax}{x \in \R^n}$. For any vector space $V \in \R^n$, we use $\Proj_V$ to denote the linear operator which maps a vector $x$ to the projection of $x$ in the subspace $V$ i.e., the nearest vector to $x$ in $V$ in terms of Euclidean distance. If the columns of $X$ are an orthonormal basis for $V$, then $\Proj_{V} = X\T{X}$.

Let $A = U\Sigma\T{V}$ be the singular value decomposition (SVD) of $A$ and let $\sigma_1 \ge \cdots \ge \sigma_{n}$ (recall $m \ge n$) denote the singular values of $A$. For $k \le n$, let $A_k \coloneqq \sum_{i=1}^k \sigma_i U_{*i} (\T{V})_{i*}$ be the matrix obtained by truncating the SVD of $A$ to the top $k$ singular values. 

We use $\frnorm{A}$ to denote the Frobenius norm $(\sum_{i,j}A_{i,j}^2)^{1/2}$ and $\opnorm{A}$ to denote the operator norm $\max_{x \ne 0}\opnorm{Ax}/\opnorm{x}$. For $p \ge 1$, we define $\lp{A} = (\sum_{i=1}^{n} \sigma_i^p)^{1/p}$ to be the Schatten-$p$ norm. As $\lp{ \cdot }$ defines a norm, we have $\lp{A+B} \le \lp{A} + \lp{B}$ for any two $m \times n$ matrices $A$ and $B$. Additionally, we have $\lp{\T{A}} = \lp{A}$ and for any unitary matrices $U', V'$, we have $\lp{U'AV'} = \lp{A}$. 
\subsection{Fast Rectangular Matrix Multiplication}\label{subsec:fast-multi}
Let $\omega$ denote the best matrix multiplication exponent. The current upper bound on $\omega$ is $\approx 2.371$ \cite{duan2022faster}
and for $\gamma < 1$, let $\omega(\gamma)$ denote the exponent such that the product of an $n \times n$ with an $n \times n^{\gamma}$ matrix can be computed using $O(n^{\omega(\gamma) + \eta})$ arithmetic operations for any constant $\eta > 0$. There exists $\alpha > 0.31$ \cite{gall2024faster, gall2018improved} such that for all $\gamma < \alpha$, $\omega(\gamma) = 2$ and for all $\gamma \ge \alpha$, 
\begin{align*}
	\omega(\gamma) \le 2 + (\omega-2)\frac{\gamma - \alpha}{1-\alpha}.
\end{align*}
See \cite{le2012faster, lotti1983asymptotic} for the above bound on $\omega(\gamma)$. Recall $\beta \coloneqq \frac{\omega-2}{1-\alpha}$. We now observe that $n^{1 - \alpha}n^2 \ge n^{\omega}$ since a matrix product of the form $[n,n,n]$ can be computed using $n^{1-\alpha}$ matrix products of the form $[n, n, n^{\alpha}]$. Hence, $1-\alpha \ge \omega-2$, which implies $\beta \le 1$.

\section{Schatten-\texorpdfstring{$p$}{p} LRA using Fast Matrix Multiplication}\label{sec:main-algorithm}
\begin{small}
    \begin{algorithm}
\caption{Block Krylov Iteration Algorithm \cite{musco2015randomized}}\label{alg:block-krylov}
\KwIn{An $n \times n$ matrix $A$, rank parameter $k$, block size $b$ and number $q$ of iterations}	
\KwOut{An orthonormal matrix $Z \in \R^{n \times k}$}
\DontPrintSemicolon
$\Pi \sim \calN(0,1)^{n \times b}$\;
$K \gets \begin{bmatrix}
A\Pi & (A\T{A})A\Pi	& \cdots & (A\T{A})^{q}A\Pi
\end{bmatrix}$ \tcp*{The Krylov Matrix}
Orthonormalize columns of $K$ to get an $n \times qb$ matrix $Q$\;
Compute $M \coloneqq \T{Q}A\T{A}Q$\;
Set $\overline{U}_k$ to the top $k$ singular vectors of $M$\;
\Return{ $Z = Q\overline{U}_k$}\;
\end{algorithm}
\end{small}

\subsection{Block Krylov Iteration Algorithm}
The block Krylov Iteration algorithm of Musco and Musco \cite{musco2015randomized} is stated as Algorithm~\ref{alg:block-krylov}. For any $b$, let $T(n,b)$ be the time to multiply an $n \times n$ matrix with an $n \times b$ matrix. The Block Krylov iteration algorithm with rank parameter $k$, block size $b \ge k$ and iteration count $q$ (with $bq \le n$) runs in time at most
$
	(2q+1)T(n, b) + n(qb)^{\omega-1} + 3T(n, qb) + (qb)^{\omega} + T(n, k).
$\footnote{Assuming that SVD of the $qb \times qb$ matrix $M$ in Algorithm~\ref{alg:block-krylov} can be computed in time $O((qb)^{\omega})$.}

Using the fact that $T(n, qb) \le qT(n,b)$ and $qb \le n$, we obtain that the time complexity of the algorithm is $O(qT(n,b)+ n(qb)^{\omega-1})$. We now have $T(n, b) \ge (b/n) n^{\omega}$ since otherwise the matrix product of the form $[n, n, n]$  can be computed quicker than in $n^{\omega}$ time by computing the $n/b$ products of the form $[n, n, b]$. Hence, $qT(n, b) \ge qb n^{\omega-1} \ge n(qb)^{\omega-1}$ using $qb \le n$. Thus,  we obtain that the time complexity of the Block Krylov Iteration algorithm with parameters $k, b, q$ satisfying $b \ge k$ and $bq \le n$ is $O(qT(n,b))$. We now state a few properties of the Block Krylov algorithm that we use throughout the paper.
\begin{theorem}
With a large probability over the Gaussian matrix $\Pi$, the following properties hold for the matrix $Z$ computed by Algorithm~\ref{alg:block-krylov}:
\begin{enumerate}
    \item There is a universal constant $c$ such that for all $i \in [k]$,
\[
        \sigma_i(\T{Z}A)^2 \ge \opnorm{\T{A}(Z)_{*i}}^2 \ge \sigma_i^2 - ({c \log^2 n}/{q^2})\sigma_{k+1}^2.
\]
    This follows from the per-vector error guarantee of Theorem~1 in \cite{musco2015randomized}.
    \item If $\gap \coloneqq (\sigma_k/\sigma_{b+1}) - 1$ and $q \ge C\log(n/\varepsilon)/\sqrt{\min(1, \gap)}$ for a large enough constant $C$, then for all $i \in [k]$,
$
        \sigma_i(\T{Z}A)^2 \ge \opnorm{\T{A}(Z)_{*i}}^2 \ge \sigma_i^2 - \varepsilon\sigma_{k+1}^2.
$
\end{enumerate}
\label{thm:block-krylov}
\end{theorem}
    The second guarantee in the above theorem follows from the gap-dependent error bounds in Theorem~11 in \cite{musco2015randomized}. Note the logarithmic dependence of $q$ on $1/\varepsilon$.
\begin{algorithm}
\caption{Schatten-$p$ Norm Subspace Approximation}\label{alg:schatten-norm-subspace-apx}
\DontPrintSemicolon
\KwIn{An $n \times n$ matrix $A$, rank parameter $k$ and an accuracy parameter $\varepsilon$}
\KwOut{Approximate Solution to the Schatten-$p$ Norm Subspace Approximation problem}
$q \gets \begin{cases}\sqrt{p} & k \le \varepsilon \cdot n^{\alpha} \\
\max(\sqrt{p}, p^{\frac{1}{2(1+2\beta)}}(k/n^{\alpha} \varepsilon)^{\frac{\beta}{1+2\beta}}) & \varepsilon \cdot n^{\alpha} \le k \le n^{\alpha}\\
\max(\sqrt{p}, {p^{\frac{1}{2(1+2\beta)}}}/{\varepsilon^{\frac{\beta}{1+2\beta}}}) & k \ge n^{\alpha}
\end{cases}$\;
$b' \gets \ceil{(3/2)\max(1, k/q^2\varepsilon)}$\;
$Z_1 \gets \textsc{BlockKrylov}(A, \text{rank} =k, \text{block size} = k, \text{iterations} = O(q\log (n))$\;
$Z_2 \gets \textsc{BlockKrylov}(A, \text{rank} =k, \text{block size} = b'+k, \text{iterations} = O(\sqrt{p}\log (n/\varepsilon))$\;
$W_1 \gets \text{colspan}(\T{A}Z_1)$\;
$W_2 \gets \text{colspan}(\T{A}Z_2)$\;
$W \gets $ $W_2$ if $\hat{\sigma}_k \ge (1+1/2p)\hat{\sigma}_{b'+k}$ and $W_1$ otherwise\tcp*{These approximations to $\sigma_k$ and $\sigma_{b'+k}$ can be computed using the $M$ matrix computed in Algorithm~\ref{alg:block-krylov}}
\end{algorithm}
\subsection{Main Theorem}
\begin{theorem}
	Given an $n\times n$ matrix $A$, a rank parameter $k$ and an accuracy parameter $\varepsilon$, Algorithm~\ref{alg:schatten-norm-subspace-apx} outputs a $k$ dimensional orthonormal matrix $W$ that with probability $\ge 0.9$ satisfies,
$
		\lp{A(I-W\T{W})} \le (1+O(\varepsilon))\lp{A - A_k}.
$
    For any constant $\eta > 0$, the running time of the algorithm is as follows:
    \begin{enumerate}
        \item For $k \le \varepsilon n^{\alpha}$, the algorithm runs in time $\tilde{O}(\sqrt{p}n^{2+\eta})$.
        \item For $\varepsilon n^{\alpha} \le k \le n^{\alpha}$, the algorithm runs in time $\tilde{O}(\max(\sqrt{p}n^{2+\eta}, p^{\frac{1}{2(1+2\beta)}}n^{2+\eta} (k/n^{\alpha}\varepsilon)^{\beta/(1+2\beta)}))$.
        \item For $k \ge n^{\alpha}$, the algorithm runs in time $\tilde{O}((p^{1/2}\varepsilon^{-\beta})^{1/(1+2\beta)}n^{2+\eta-\alpha\beta}k^{\beta})$.
    \end{enumerate}
	\label{thm:final-theorem-our-algorithm}
\end{theorem}
Assuming $p$ is a constant independent of $\varepsilon$, the dependence on $\varepsilon$ is at least better than $\varepsilon^{-1/3}$ as $\beta \le 1$ which implies $\beta/(1+2\beta) \le 1/3$. The proof of this theorem is similar to that of \cite{bakshi2022low}. We include the proof in the appendix.

\section{Comparison with the Algorithm of Li and Woodruff \texorpdfstring{\cite{li2020input}}{LW20}}\label{sec:comparison}
For $n \times n$ matrices and $p > 2$, the algorithm of \cite{li2020input} for the Schatten-$p$ norm Subspace Approximation problem, shown in Algorithm~\ref{alg:original-algorithm-li-woodruff} runs in time
\begin{align}
	O(n^2 \log n) + \tilde{O}_p\left(\frac{n^{\omega (1-2/p)}k^{2\omega/p}}{\varepsilon^{2\omega/p+2}} + n^{1+(\omega-1)(1-2/p)}(k/\varepsilon)^{2(\omega-1)/p}\right).
	\label{eqn:original-running-time}
\end{align}    
Let $K = k + \varepsilon/\eta_1 = k + n^{1-2/p}k^{2/p}/\varepsilon^{2/p}$. To obtain the above running time, they use a ridge leverage score sampling algorithm to compute a matrix $S$ with $s = O(\varepsilon^{-2}K\log n)$ rows that satisfies \eqref{eqn:mixed-guarantee} with a large probability. The same guarantee can instead be obtained by using the Sub-sampled Randomized Hadamard Transform (SRHT) \cite{tropp2011improved} with $s = O(\varepsilon^{-2}K\log n)$ rows and the matrix-product $SA$ can be computed in time $O(n^2 \log n)$. To obtain the subspace embedding guarantee for $T$ as required in Algorithm~\ref{alg:original-algorithm-li-woodruff}, we can let the matrix $T$ again be an SRHT with $r = O(\varepsilon^{-2}s\log n)$ columns and the product $SAT$ can be computed in time $O(ns\log s) = O(\varepsilon^{-2}n(k + n^{1-2/p}k^{2/p}/\varepsilon^{2/p}))$.

The singular value decomposition of the matrix $SAT$ can be computed in $O(rs^{\omega-1}) = O(\varepsilon^{-2\omega}(k+n^{1-2/p}k^{2/p}/\varepsilon^{2/p})^{\omega}\polylog(n))$ time and a basis for the rowspace of $\T{W}SA$ can be computed in $O(skn)$ time. Overall, for constant $k$ and $\varepsilon$, the algorithm of \cite{li2020input} runs in time
$
	\tilde{O}(n^2 + (n^{1-2/p})^{\omega}).
$
For $p > 2\omega/(\omega-2)$, their algorithm runs in $n^{2+c_p}$ time for a constant $c_p > 0$ that depends on $p$. For the same parameters, our algorithm runs in $\tilde{O}(n^2)$ time and hence we have an improvement. For $k \le n^{\alpha}$ and $\varepsilon = 1/n$, their algorithm runs in time $\Omega(n^{\omega})$ which means that computing the SVD of $A$ is already faster whereas our algorithm runs in time $\tilde{O}(n^{2 + \frac{(1-\alpha)\beta}{1+2\beta}}) = o(n^{\omega})$ if $\omega > 2$. Hence, our algorithm improves upon the algorithm of \cite{li2020input} for a wide range of parameters. We note that computing the SVD of $SAT$ turns out to be the most expensive step for large $p$. In the next subsection, we show that our Algorithm~\ref{alg:schatten-norm-subspace-apx} can be used to sidestep the computation of the SVD of $SAT$, thereby giving an even faster algorithm.

We call $p^* = 2\omega/(\omega-2)$, the \emph{crossover} point. For $p > p^*$, our Algorithm~\ref{alg:schatten-norm-subspace-apx} is faster than the algorithm of \cite{li2020input}. For the current value of $\omega \approx 2.37$, $p^* \approx 12.8$. For $p < p^*$, the leading order term in the time complexity of Algorithm~\ref{alg:original-algorithm-li-woodruff} is $O(n^2 \log n)$ for $\varepsilon > n^{-c_p}$ for a constant $c_p$ depending on $p$, and hence is faster than Algorithm~\ref{alg:schatten-norm-subspace-apx}.
\begin{algorithm}
\caption{Schatten-$p$ Norm Low Rank Approximation for $p > 2$ \cite{li2020input}}\label{alg:original-algorithm-li-woodruff}
\KwIn{A matrix $A \in \R^{m \times n}$ and an accuracy parameter $\varepsilon$}
\KwOut{A rank-$k$ orthonormal projection $Q$ satisfying $\lp{A(I-Q)} \le (1+\varepsilon)\lp{A - A_k}$}	
\DontPrintSemicolon
$\eta_1 \gets O(\varepsilon^{1+2/p}/k^{2/p}n^{1-2/p})$\;
$S$ be a matrix with $s$ rows that satisfies
\begin{align}
	(1-\varepsilon)\T{A}A - \eta_1\frnorm{A-A_k}^2 \cdot I  \preceq \T{A}\T{S}SA \preceq (1+\varepsilon)\T{A}A + \eta_1\frnorm{A-A_k}^2 \cdot I .
	\label{eqn:mixed-guarantee}
\end{align}\;    
\vspace{-2em}
$T \gets $ Subspace embedding for $s$-dimensional subspaces with error $O(\varepsilon)$\;
$W \gets$ Top $k$ left singular vectors of $SAT$\;
$Z \gets$ Matrix whose columns are an orthonormal basis for the row space of $\T{W}SA$\;
$Q \gets Z\T{Z}$
\end{algorithm}
\subsection{Further Improving the running time of \texorpdfstring{\cite{li2020input}}{LW20} using our algorithm}
Given an $n \times n$ matrix $A$, $p \ge 1$ and $r \le n$, let $\|A\|_{(p,r)} = (\sum_{i=1}^r \sigma_i(A)^p)^{1/p}$. We can show that $\| \cdot \|_{(p,r)}$ is a norm over $n \times n$ matrices. As $\|\cdot\|_{(p,r)} $ is unitarily invariant, we have by Eckart-Young-Mirsky's theorem that
$
	\|A - A_k\|_{(p,r)} = \min_{\text{rank-}k\, B}\|A-B\|_{(p,r)}.
$
In Lemma~4.2 of \cite{li2020input}, they show that for $S$ satisfying \eqref{eqn:mixed-guarantee}, if $\hat{Q}$ is a rank-$k$ projection matrix with
\begin{align}
	\|SA(I-\hat{Q})\|_{(p,r)} \le (1+\varepsilon)\min_{\substack{\text{rank-}k\\ \text{projections}\, Q}}\|SA(I-Q)\|_{(p,r)},
	\label{eqn:required-guarantee-truncated-norm}
\end{align}
then
$
	\lp{A(I-\hat{Q})}^p \le (1+C_p \varepsilon)\lp{A-A_k}^p,
$
for a constant $C_p$ that only depends on $p$. They show that the matrix $Q$ returned by Algorithm~\ref{alg:original-algorithm-li-woodruff} satisfies \eqref{eqn:required-guarantee-truncated-norm} and then conclude that the matrix $Q$ is a $1+O(\varepsilon)$ approximation to the Schatten-$p$ norm low rank approximation problem. We will now argue that there is a faster algorithm for computing a projection that satisfies \eqref{eqn:required-guarantee-truncated-norm}. The algorithm does not require the computation of the SVD of the matrix $SAT$ and hence does not incur the $O_{p,k,\varepsilon}(n^{(1-2/p)\omega})$ term in the running time. We first show that a $1+\varepsilon$ approximate solution to the Schatten-$p$ norm subspace approximation problem, is a $1+\varepsilon n/r$ approximation to the $(p,r)$ subspace approximation problem. 
\begin{lemma}
For an arbitrary $m \times n$ matrix $A$ ($m \le n$), if $\hat{Q}$ is a rank-$k$ projection matrix satisfying
$
	\lp{A(I-\hat{Q})}^p \le (1+\varepsilon)\lp{A-A_k}^p
$ and $\text{colspan}(\hat{Q}) \subseteq \text{rowspan}(A)$,
then for any $r \le n$,
\[
	\|{A(I-\hat{Q})}\|_{(p,r)}^p \le (1 + \varepsilon\left\lceil{{(m-k)}/{r}}\right\rceil)\|A-A_k\|_{(p,r)}^p.
\]
\label{lma:truncated-to-full-reduction}
\end{lemma}
\begin{proof}
	Let $\hat{Q}$ be a rank-$k$ projection such that
	\begin{align*}
		\lp{A(I-\hat{Q})}^p \le (1+\varepsilon)\min_{\text{rank-}k\, \text{projections}\, Q}\lp{A(I-Q)}^p = (1+\varepsilon)\sum_{i=k+1}^n \sigma_i(A)^p.
	\end{align*}
	Note that $\lp{A(I-\hat{Q})}^p = \sum_{i=1}^{m-k}\sigma_i(A(I-\hat{Q}))^p$ since the matrix $A(I-\hat{Q})$ has rank at most $m-k$ from our assumption that $\text{colspan}(Q) \subseteq \text{rowspan}(A)$. Now, $\|A(I-\hat{Q})\|_{(p,r)}^p = \sum_{i=1}^{r}\sigma_i(A(I-\hat{Q}))^p$ and therefore,
	\begin{align*}
		\|A(I-\hat{Q})\|_{(p,r)}^p &= \|A(I-\hat{Q})\|_{S_p}^p - \sum_{i=r+1}^{m-k}\sigma_i(A(I-\hat{Q}))^p\\
  &\le (1+\varepsilon)\sum_{i=k+1}^m \sigma_i(A)^p - \sum_{i=r+1}^{m-k}\sigma_i(A(I-\hat{Q}))^p.
	\end{align*}
	Since the matrix $A\hat{Q}$ has rank at most $k$, by Weyl's inequality, 
$
		\sigma_i(A(I-\hat{Q})) \ge \sigma_{i+k}(A)
$
which implies
	\begin{align*}
		\|A(I-\hat{Q})\|_{(p,r)}^p &\le \sum_{i=k+1}^{k+r}\sigma_i(A)^p + \varepsilon\lp{A-A_k}^p  + \left(\sum_{i=k+r+1}^m \sigma_i(A)^p - \sum_{i=r+1}^{m-k}\sigma_i(A(I-\hat{Q}))^p\right)\\
		&\le \min_{\text{rank-}k\, \text{projections}\, Q}\|A(I-Q)\|_{(p,r)}^p + \varepsilon\lp{A-A_k}^p.
	\end{align*}
	Finally, using the fact that $\|A-A_k\|_{S_p}^p \le \ceil{(m-k)/r}\|A-A_k\|_{(p,r)}^p$, we obtain
	\begin{align*}
		\|A(I-\hat{Q})\|_{(p,r)}^p \le (1 + \varepsilon\ceil{(m-k)/r})\|A-A_k\|_{(p,r)}^p. &\qedhere
	\end{align*}
\end{proof}
Finally, we have the following lemma which shows how to find an approximate solution to the $(p,r)$ Low Rank Approximation problem.
\begin{lemma}
Let $A \in \R^{m \times n}$	be an arbitrary matrix with $m \le n$. Given parameters $k$, $p$, $r$ and $\varepsilon$, there is a randomized algorithm to find a rank-$k$ projection $\hat{Q}$, that with probability $\ge 9/10$ satisfies,
\[
	\|A(I-\hat{Q})\|_{(p,r)}^p \le (1+\varepsilon)\|A - A_k\|_{(p,r)}^p.
\]
For constant $p$ and $k \le m^{\alpha}$ and any constant $\eta > 0$, the randomized algorithm runs in time \[\tilde{O}(m^{2+\eta+(1-\alpha)\beta/(1+2\beta)}k^{\beta/(1+2\beta)}\poly(1/\varepsilon) + nm + nk^{\omega-1})\] and for $k \ge m^{\alpha}$, the algorithm runs in 
\[
\tilde{O}(m^{2 + \eta -\alpha\beta + \frac{\beta}{1+2\beta}}k^{\beta}\poly(1/\varepsilon) + nm^{1-\alpha\beta}k^{\beta} + nk^{\omega-1})\] time.
\label{lma:full-lemma-truncated}
\end{lemma}
\begin{proof}
First we note that
\begin{align*}
	\min_{\text{rank-}k\, \text{projections}\, Q}\|A(I-Q)\|_{(p,r)}^p = \min_{\text{rank-}k\, \text{projections}\, W}\|(I-W)A\|_{(p,r)}^p = \|A-A_k\|_{(p,r)}^p.
\end{align*}
Let $T$ be an SRHT matrix with $O(\varepsilon^{-2}m\polylog(n))$ rows. With a large probability, $T$ is an $\varepsilon$ subspace embedding for the rowspace of matrix $A$. Then
\begin{align*}
	(1-\varepsilon)A\T{A} \preceq AT\T{T}\T{A} \preceq (1+\varepsilon)A\T{A}
\end{align*}
and further for all rank-$k$ projections $W$, 
\begin{align*}
	(1-\varepsilon)(I-W)A\T{A}(I-W) \preceq (I-W)AT\T{T}\T{A}(I-W) \preceq (1+\varepsilon)(I-W)A\T{A}(I-W).
\end{align*}
We then have for all $i$ that $\sigma_i((I-W)AT) = (\sqrt{1 \pm \varepsilon})\sigma_i((I-W)A)$. Therefore, $\|(I-W)AT\|_{(p,r)}^p = (1 \pm \varepsilon)^{p/2}\|(I-W)A\|_{(p,r)}^p$ for all rank-$k$ projections $W$. Let Algorithm~\ref{alg:schatten-norm-subspace-apx} be run on the matrix $\T{T}\T{A}$ with rank parameter $k$ and approximation parameter $\varepsilon/pm$. By Theorem~\ref{thm:final-theorem-our-algorithm}, we obtain a rank-$k$ projection $\widehat{W}$ satisfying
\begin{align*}
	\|\T{T}\T{A}(I-\widehat{W})\|_{p} \le (1+\varepsilon/pm)\min_{\text{rank-}k\, \text{projections}\, W}\|\T{T}\T{A}(I-W)\|_{p}
\end{align*}
Using, Lemma~\ref{lma:truncated-to-full-reduction}, we obtain that
\begin{align*}
	\|\T{T}\T{A}(I-\widehat{W})\|_{(p,r)}^p \le (1+\varepsilon)\min_{\text{rank-}k\, \text{projections}\, W}\|\T{T}\T{A}(I-W)\|_{(p,r)}^p.
\end{align*}
By using the relation between $\|(I-W)AT\|_{(p,r)}^p$ and $\|(I-W)A\|_{(p,r)}^p$ for all projections $W$, we get
\begin{align*}
	\|\T{A}(I-\widehat{W})\|_{(p,r)}^p &\le \frac{(1+\varepsilon)^{p/2+1}}{(1-\varepsilon)^{p/2}}\min_{\text{rank-}k\, \text{projections}\, W}\|\T{A}(I-W)\|_{(p,r)}^p\\
	&\le (1 + O(\varepsilon p))\|A - A_k\|_{(p,r)}^p.
\end{align*}
Now, $\|A - A({\widehat W}A)^{+}({\widehat W}A)\|_{(p,r)}^p \le \|A - {\widehat W}A\|_{(p,r)}^p = \|(I - \widehat W){A}\|_{(p,r)}^p \le (1+O(\varepsilon p))\|A-A_k\|_{(p,r)}^p$. Scaling $\varepsilon$, we obtain the result.

\paragraph*{Runtime Analysis.} The matrix $AT$ can be computed in time $O(mn\log n)$. For constant $p$, Algorithm~\ref{alg:schatten-norm-subspace-apx} runs on the matrix $\T{T}A$ in time $\tilde{O}(m^{2+\eta+(1-\alpha)\beta/(1+2\beta)}k^{\beta/(1+2\beta)}\poly(1/\varepsilon))$ for $k \le m^{\alpha}$ and when $k \ge m^{\alpha}$, the algorithm runs in time $\tilde{O}(m^{2+\eta-\alpha\beta+\frac{\beta}{1+2\beta}}k^{\beta}\poly(1/\varepsilon))$. Finally, the rowspace of $\T{\widehat{W}}A$ can be computed in time $O(nm + nk^{\omega-1})$ for $k \le m^{\alpha}$ and $O(nm^{1-\alpha\beta}k^{\beta} + nk^{\omega-1})$ for $k \ge m^{\alpha}$.
\end{proof}

Using the above lemma, we can find a rank-$k$ projection $\hat{Q}$ that satisfies
\[
	\|SA(I-\hat{Q})\|_{(p,r)}^p \le (1+\varepsilon)\|A - A_k\|_{(p,r)}^p
\]
in time $\tilde{O}((n^{1-2/p})^{2+\eta+(1-\alpha)\beta/(1+2\beta)}\poly(1/\varepsilon
) + n^2)$ for constant $k$ improving on the $\tilde{O}(n^2 + (n^{(1-2/p)})^{\omega}\poly(1/\varepsilon))$ running time of \cite{li2020input} for the current value of $\omega$ since
$
    2 + \frac{(1-\alpha)\beta}{1+2\beta} = 2 + \frac{\omega-2}{1+2\beta} < \omega
$
if $\beta \ne 0$. We thus have the following theorem.
\begin{theorem}
	Given a dense $n \times n$ matrix $A$, a constant rank parameter $k$ and any constant $\eta > 0$, there is a randomized algorithm that runs in time $\tilde{O}((n^{1-2/p})^{2+\eta+(1-\alpha)\beta/(1+2\beta)}\poly(1/\varepsilon
) + n^2)$ and outputs a rank-$k$ projection $\hat{Q}$ that, with probability $\ge 9/10$, satisfies
$
		\|A(I-\hat{Q})\|_{S_p}^p \le (1+\varepsilon)\|A-A_k\|_{S_p} ^p.
$
 \label{thm:combination}
\end{theorem}
For this algorithm, the crossover point is $\tilde{p} = \frac{4(1+2\beta)}{\omega-2} + 2$ i.e., only when $p > \tilde{p}$, Algorithm~\ref{alg:schatten-norm-subspace-apx} is faster than the algorithm in the above theorem for constant $k$ and $\varepsilon$. For current values of $\omega, \alpha$, we have $\tilde{p} \approx 22$. In particular, for constant $k$ and $\varepsilon > n^{-c_p}$, for $p \lesssim 22$, the algorithm has a time complexity of only $\tilde{O}(n^2)$.

\section{Stability of LazySVD}\label{sec:lazysvd}
\subsection{Finite Precision Preliminaries}
Following the presentation of \cite{stability-of-lanczos}, we say that a floating point machine has precision $\varepsilon_{\machine}$ if it can perform computations to relative error $\varepsilon_{\machine}$. More formally, let $\fl(x \circ y)$ be the result of the computation $x \circ y$ on the floating point machine where $\circ \in \set{+, -, \times, \div}$. We say that the floating point machine has a precision $\varepsilon_{\machine}$ if for all $x$ and $y$, $\fl(x \circ y) = (1+\delta)(x \circ y)$ where $|\delta| \le \varepsilon_{\machine}$. Additionally, we also require $\fl(\sqrt{x}) = (1+\delta)\sqrt{x}$ for some $\delta$ with $|\delta| \le \varepsilon_{\machine}$. Ignoring overflow or underflow, a machine which implements the IEEE floating point standard with $\ge \log_2(1/\varepsilon_{\machine})$ bits of precision satisfies the above requirements (see \cite[Section~5]{stability-of-lanczos}). Given matrices $A$ and $B$ with at most $n$ rows and columns, we can compute a matrix $C$, on a floating point machine, that satisfies $\opnorm{C - A \cdot B} \le \varepsilon_{\machine}\poly(n)\opnorm{A}\opnorm{B}$ by directly computing $C_{ij}$ as $\fl(\sum_{k} A_{ik}B_{kj})$. 
\subsection{Stability Analysis}
\begin{algorithm}
    \caption{LazySVD \cite{allen2016lazysvd}}\label{alg:lazysvd}
    \KwIn{A positive semidefinite matrix $M \in \R^{d \times d}, k \le d, \varepsilon, \varepsilon_{\pca}, \eta$}
    \KwOut{Vectors $v_1,\ldots,v_k$}
    \DontPrintSemicolon
    $M_0 \gets M$ and $V_0 \gets []$\;
    \For{$s=1,\ldots,k$}{
        $v_s' \gets \AppxPCA_{\varepsilon/2, \varepsilon_{\pca}, \eta/k}(M_{s-1})$\;
        $v_s \gets (I-V_{s-1}\T{V_{s-1}})v_s'/\opnorm{(I-V_{s-1}\T{V_{s-1}})v_s'}$\;
        $V_s \gets [V_{s-1}\, v_s]$\;
        $M_s \gets (I-V_s\T{V_s})M(I-V_s\T{V_s})$ \tcp*{The matrix $M_s$ is not computed as we only need matrix vector products with $M_s$}
    }
    \Return $V_k$\;
\end{algorithm}
The LazySVD algorithm (Algorithm~\ref{alg:lazysvd}) of \cite{allen2016lazysvd} crucially requires a routine called $\AppxPCA$ that computes an approximation to the top eigen vector of the given positive semidefinite matrix. While they use a particular $\AppxPCA$ algorithm in their results, any routine that satisfies the following definition can be plugged into the LazySVD algorithm. 
\begin{definition}[\AppxPCA]\label{dfn:appxpca}
We say that an algorithm is $\AppxPCA$ with parameters $\varepsilon, \varepsilon_{\pca}$ and $\eta$ if given a positive semidefinite matrix $M \in \R^{d \times d}$ with an orthonormal set of eigenvectors $u_1,\ldots,u_d$ corresponding to eigenvalues $\lambda_1 \ge \cdots \ge \lambda_d \ge 0$, the algorithm outputs a unit vector $w$ such that with probability $\ge 1- \eta$,
$
    \sum_{i \in [d]: \lambda_i \le (1-\varepsilon)\lambda_1} \la w, u_i\ra^2 \le \varepsilon_{\pca}.
$
\end{definition}
We now show that Lanczos algorithm can be used to stably compute a vector that satisfies the $\AppxPCA$ guarantee.
\begin{lemma}
    If for any vector $x$, we can compute a vector $y$ such that \[\opnorm{y - M_sx} \le O(\varepsilon_{\machine}\poly(n)\kappa)\opnorm{M_s}\opnorm{x}\] and if $\varepsilon_{\machine} \le \poly(\varepsilon_{\pca}\eta/n\kappa)$, then we can compute a unit vector $v$ such that with probability $\ge 1 - \eta$, 
$
        \sum_{i: \lambda_i(M_s) \le (1-\varepsilon)\lambda_1(M_s)} \la v, u_i(M_s)\ra^2 \le \varepsilon.
$
    The algorithm uses $O(\frac{1}{\sqrt{\varepsilon}}\poly(\log(d/\varepsilon\eta\varepsilon_{\pca})))$ matrix vector products with $M_s$.
    \label{lma:apxpca}
\end{lemma}
\begin{proof}
     Let $\bz$ be a $d$ dimensional random vector with each coordinate being an independent Gaussian random variable. Let $M_s = \sum_i \lambda_i u_i \T{u_i}$ be the eigendecomposition. Let $r$ be the largest index such that $\lambda_r \ge (1-\varepsilon)\lambda_1$. Consider the vector $M_s^q \bz$ for a $q$ we choose later. We have
    \begin{align*}
        \by = M_s^q \bz = \sum_{i=1}^d \lambda_i^q \la u_i, \bz\ra u_i.
    \end{align*}
    Consider $\la u_1, \bz\ra$. By 2-stability of Gaussian random variables, $\la u_1, \bz\ra \sim N(0, \opnorm{u_1}^2) = N(0,1)$. Hence with probability $1 - \eta$, $|\la u_1, \bz\ra| \ge \eta$. We also have that with probability $\ge 1 - \eta$, for all $i=1,\ldots,d$ $|\la u_i, \bz\ra| \le O(\sqrt{\log d/\eta})$. Condition on these events. Now,
    $
        \opnorm{\by}^2 = \sum_{i=1}^d \lambda_i^{2q} \la u_i, \bz\ra^2 \ge \lambda_1^{2q}\la u_1, \bz\ra^2 \ge \lambda_1^{2q}\eta^2.
    $
    Define $\hat{\by} = \by/\opnorm{\by}$. Let $i > r$ so that $\lambda_i < (1-\varepsilon)\lambda_1$ by definition of $r$. We have
    \begin{align*}
        |\la u_i, \hat{\by}\ra| = \frac{|\la u_i, \by\ra|}{\opnorm{\by}} \le \frac{\lambda_i^{q}|\la u_1, \bz\ra|}{\lambda_1^q \eta} \le \frac{\lambda_i^q \sqrt{\log d/\eta}}{\lambda_1^q \eta} \le (1-\varepsilon)^q \frac{C\sqrt{\log d/\eta}}{\eta}.
    \end{align*}
    If $q \ge C\varepsilon^{-1}\log(d/\varepsilon_{\pca}\eta)$ for a large enough constant $C$, we get $|\la u_i, \hat{\by}\ra| \le \poly(\varepsilon_{\pca}/d)$. Thus, $\sum_{i=r+1}^d |\la u_i, \hat{\by}\ra|^2 \le \poly(\varepsilon_{\pca})$.

    Now define $f(x) = x^q$ so that $f(M_s)\bz = \by$ and define $\rho = \lambda_1/q$. From \cite[Chapter~3]{sachdeva2014faster} there is a polynomial $p(x)$ of degree $\sqrt{2q\log 1/\gamma}$ such that for all $x \in [-\rho, \lambda_1 + \rho]$, 
    \begin{align*}
        |p(x) - x^q| \le e\gamma \lambda_1^q. 
    \end{align*}
    As we can compute matrix-vector products with $M_s$ up to an additive error of $O(\varepsilon_{\machine}\poly(n)\kappa)$, using Theorem~1 of \cite{stability-of-lanczos} as long as $\varepsilon_{\machine} \le \varepsilon'\rho/(\poly(n)\kappa\opnorm{M_s}) \le \varepsilon'/\poly(n)\kappa$, we can compute a vector $\by'$ on a floating point machine, using $\sqrt{2q\log 1/\gamma}$ iterations such that
    \begin{align*}
        \opnorm{\by - \by'} = \opnorm{(M_s)^q \bz - \by'}&\le ((7e\gamma \sqrt{2q\log 1/\gamma}) \lambda_1^q + \varepsilon'\lambda_1^q)\opnorm{\bz}\\ &\le O(\gamma\sqrt{2q\log 1/\gamma} + \varepsilon')\lambda_1^q \sqrt{d}.
    \end{align*}
    where we used that $\opnorm{\bz} \le O(\sqrt{d})$ with high probability. As $\opnorm{\by} \ge \lambda_1^q \eta$, we further obtain that
    \begin{align*}
        \opnorm{\by - \by'} \le O(\gamma \sqrt{2q\log 1/\gamma} + \varepsilon') \sqrt{d}\opnorm{\by}/\eta.
    \end{align*}
    We set $\gamma = \poly(\varepsilon_{\pca}\eta/dq)$ and $\varepsilon' = \poly(\varepsilon_{\pca}\eta/d)$ to obtain that $\opnorm{\by - \by'} \le \poly(\varepsilon_{\pca}/d)\opnorm{\by}$. Thus,
    \begin{align*}
        \opnorm{\hat{\by} - \by'/\opnorm{\by'}} \le \opnorm{\by/\opnorm{\by} - {\by'}/\opnorm{\by'}} \le \poly(\varepsilon_{\pca}/d).
    \end{align*}
    On a floating point machine, we can normalize the vector $\by'$ to obtain a vector $\hat{\by}'$ such that $\opnorm{\hat{\by}'} = (1 \pm \varepsilon_{\machine}\poly(d))$ and $\opnorm{\hat{\by}' - \by'/\opnorm{\by'}} \le \varepsilon_{\machine}\poly(d)$. By triangle inequality, we then obtain $\opnorm{\hat{\by} - \hat{\by}'} \le \poly(\varepsilon/d) + \varepsilon_{\machine}\poly(d)$. Finally, for $i > r$
    \begin{align*}
        |\la u_i, \hat{\by}'\ra| \le |\la u_i, \hat{\by}\ra| + \opnorm{\hat{\by} - \hat{\by}'} \le \poly(\varepsilon_{\pca}/d) + \varepsilon_{\machine}\poly(d)
    \end{align*}
    which then implies that as long as $\varepsilon_{\machine} \le \poly(\varepsilon_{\pca}/d)$, we get $\sum_{i=r+1}^d \la u_i, \hat{\by}\ra^2 \le \poly(\varepsilon_{\pca})$.

    Thus, we overall obtain that if $\varepsilon_{\machine} \le \poly(\varepsilon_{\pca}\eta/d\kappa)$, we can obtain a vector $\hat{\by}'$ by running the Lanczos method for $O(\frac{1}{\sqrt{\varepsilon}}\poly(\log (d/\varepsilon_{\pca}\eta\varepsilon)))$ iterations such that with probability $\ge 1 - \eta$, $\opnorm{\hat{\by}'} = (1 \pm \varepsilon_{\machine}\poly(d))$ and \[\sum_{i:\lambda_i(M_s) \le (1-\varepsilon)\lambda_1(M_s)}\la \hat{\by}', u_i\ra^2 \le \varepsilon_{\pca}.\] Overall, the algorithm uses $O(\frac{1}{\sqrt{\varepsilon}}\poly(\log (d/\varepsilon\eta \varepsilon)))$ matrix vector products with $M_s$ and uses an additional $O(\frac{d}{\sqrt{\varepsilon}}\poly(\log (d/\varepsilon_{\pca}\eta \varepsilon)))$ floating point operations.
\end{proof}
Finally, we modify the LazySVD algorithm (see Algorithm~\ref{alg:modified-lazy-svd}) to make it more stable when implemented on a floating point machine. The modification preserves the semantics of the algorithm in the real number model while allowing the stability analysis to go through. For the matrices that we need to run the routine $\AppxPCA$ on, we show that we can compute very accurate matrix-vector products so that the Lanczos algorithm can be used to approximate the top eigenvector to obtain the following theorem:
\begin{theorem}
    Given an $n \times d$ matrix $A$ with condition number $\kappa(A) = \sigma_1(A)/\sigma_{k+1}(A)$, an accuracy parameter $\varepsilon$, a rank parameter $k$ and probability parameter $\eta$, if $\varepsilon_{\machine} \le \poly(\varepsilon\eta/n\kappa(A))$, there is an algorithm that outputs a $d \times k$ matrix $V_k$ such that $\kappa(V_k) \le 4$ and for all $p \in [2, \infty]$,
        \[
        \|A(I - \Proj_{\colspace(V_k)})\|_{S_p} \le (1 + O(\varepsilon))\|A - A_k\|_{S_p}
        \]
    and runs in time $O(\frac{\nnz(A)k}{\sqrt{\varepsilon}}\poly(\log(d\kappa(A)/\varepsilon\eta)) + d\poly(k, \log(dk/\eta\varepsilon)))$.
    \label{thm:lazysvd-stability}
\end{theorem}

\begin{algorithm}
    \caption{Modified LazySVD}\label{alg:modified-lazy-svd}
    \KwIn{A positive semidefinite matrix $M \in \R^{d \times d}, k \le d, \varepsilon, \varepsilon_{\pca}, \eta$}
    \KwOut{Vectors $v'_1, \ldots, v'_k$}
    \DontPrintSemicolon
    $M_0 \gets M$ and $V_0 \gets []$\;
    \For{$s = 1,\ldots,k$}{
        $v'_s \gets \AppxPCA_{\varepsilon, \varepsilon_{\pca}, \eta/k}((I-\Proj_{\colspace(V_{s-1})})M(I - \Proj_{\colspace(V_{s-1})}))$\;
        $V_{s} \gets [V_{s-1}\, v_s']$\;
    }
    \Return $V_k$\;
\end{algorithm}
For convenience, we denote any algorithm that satisfies Definition~\ref{dfn:appxpca} as $\AppxPCA_{\varepsilon, \varepsilon_{\pca}, \eta}$. We abuse notation and say that if a unit vector $w$ satisfies $\sum_{i \in [d]:\lambda_i(M) \le (1-\varepsilon)\lambda_1(M)} \la w, u_i(M)\ra^2 \le \varepsilon_{\pca}$, then ``$w$ is $\AppxPCA_{\varepsilon, \varepsilon_{\pca}}(M)$''.

In \cite{allen2016lazysvd}, the authors show that if $\varepsilon_{\pca} = \poly(\varepsilon, 1/d, \lambda_{k+1}/\lambda_1)$, then with probability $\ge 1 - \eta$ (union bounding over the success of all $k$ calls to the $\AppxPCA$ routine), the orthonormal matrix $V_k$ output by Algorithm~\ref{alg:lazysvd} satisfies
\begin{enumerate}
    \item \label{lazysvd:prop-1} $\opnorm{(I-V_k\T{V_k})M(I-V_k\T{V_k})} \le \frac{\lambda_{k+1}(M)}{(1-\varepsilon)}$,
    \item \label{lazysvd:prop-2} $(1-\varepsilon)\lambda_k(M) \le \T{v_k}Mv_k \le \frac{1}{1-\varepsilon} \lambda_k(M)$ and
    \item \label{lazysvd:prop-3} for every $p \ge 1$, $\|(I-V_k\T{V_k})M(I-V_k\T{V_k})\|_{S_p} \le (1 + O(\varepsilon))(\sum_{i=k+1}^d \lambda_i^p)^{1/p}$.
\end{enumerate}

Since, the modified algorithm (Algorithm~\ref{alg:modified-lazy-svd}) has the same semantics as Algorithm~\ref{alg:lazysvd}, the properties \ref{lazysvd:prop-1} and \ref{lazysvd:prop-3} continue to hold for the modified LazySVD algorithm.

The advantage of the modification is that given any vector $x$, we can compute $(I - \Proj_{\colspace(V_{s}}))x$ very accurately on a floating point machine using stable algorithms for the least squares problem, thereby obtaining a vector $y$ on a floating point machine that is a very good approximation to $M_s x = (I-\Proj_{\colspace(V_s)})M((I-\Proj_{\colspace(V_s)}))x$ for any given $x$. Below we have a result that states the stability of solving the Least Squares problem on a floating point machine.
\begin{theorem}[{Theorem~19.1 of \cite{trefethen}}]
    The algorithm for solving the least squares problem $\min_x \opnorm{Ax - b}^2$ using Householder triangulation is backwards stable in the sense that the solution $\tilde{x}$ satisfies
    \begin{align*}
        \opnorm{(A + \delta A)\tilde{x} - b}^2 = \min_x \opnorm{(A + \delta A)x - b}^2
    \end{align*}
    for some matrix $\delta A$ satisfying $\opnorm{\delta A} \le O(\varepsilon_{\machine}\opnorm{A})$.
\end{theorem}
Let $x^* = A^+b$ and from the above theorem, we have $\tilde{x} = (A+\delta A)^{+}b$. Assuming $\varepsilon_{\machine} \le 1/2\kappa(A)$, we have $A+\delta A$ is full rank and therefore $(A+\delta A)^+ = (\T{(A+\delta A)}(A+\delta A))^{-1}\T{(A+\delta A)}$ using which we obtain that $\opnorm{A\tilde{x} - Ax^*} \le O(\varepsilon_{\machine}\poly(\kappa(A))\opnorm{b})$. Note that $Ax^* = \Proj_{\colspace(A)}b$. Thus, given a matrix $A$ and a vector $x$, we can compute a vector $y$ on a floating point machine such that
$
    \opnorm{y - \Proj_{\colspace(A)} x} \le O(\varepsilon_{\machine}\poly(\kappa(A), d)\opnorm{x}).
$

Finally, we can compute another vector $y'$ satisfying
$
    \opnorm{y' - (I-\Proj_{\colspace(A)})x} \le O(\varepsilon_{\machine}\poly(\kappa(A), d)\opnorm{x}).
$
Thus, given any vector $x$, if operations are computed using machine precision $\varepsilon_{\machine}$ and if we assume that for any arbitrary vector $x$, we can compute a vector $y$ satisfying $\opnorm{y - Mx} \le \varepsilon_{M}\opnorm{M}\opnorm{x}$, then given any vector $x$, we can compute a vector $y$ on a floating point machine satisfying
$
    \opnorm{y - M_s x} \le O(\varepsilon_{\machine}\poly(\kappa(V_s), d) + \varepsilon_M)\opnorm{M}\opnorm{x}.
$

We now bound $\kappa(V_s)$. Assume that the vector $v_{s}'$ satisfies $\opnorm{v_{s}'} = (1 \pm \poly(d)\varepsilon_{\machine})$. If the vector $v_{s}'$ is $\AppxPCA_{\varepsilon, \varepsilon_{\pca}}(M_{s-1})$, then
\begin{align*}
    \opnorm{\Proj_{\colspace(V_{s-1})} v_s'}^2 \le \sum_{i \in [d]: \lambda_i(M_{s-1}) \le (1-\varepsilon)\lambda_1(M_{s-1})} \la v_s', u_i(M_s)\ra^2 \le \varepsilon_{\pca}
\end{align*}
where the first inequality follows from the fact that colspace($V_{s-1}$) is spanned by the eigenvectors of $M_{s-1}$ corresponding to zero eigenvalues. Using the above inequality, we can upper bound $\sigma_{\max}(V_s)$ and lower bound $\sigma_{\min}(V_s)$. 
\begin{lemma}
    Suppose $V_{s-1}$ is a $d \times (s-1)$ matrix such that $\sigma_{\max}(V_{s-1}) = \alpha_{s-1}$ and $\sigma_{\min}(V_{s-1}) = \beta_{s-1}$. Let $v_s'$ be a vector with $\opnorm{v'_s} = (1 \pm \poly(d)\varepsilon_{\machine})$ and satisfies
$
        \opnorm{\Proj_{\colspace(V_{s-1})}v'_s}^2 \le \varepsilon_{\pca}.
$
    Let $V_s = [V_{s-1}\, v'_s]$. Then 
$
    \sigma_{\max}(V_s) \le \max(\sigma_{\max}(V_{s-1}), 1+\poly(d)\varepsilon_{\machine}) + \sqrt{\varepsilon_{\pca}}
$ and 
\begin{align*}
    \sigma_{\min}(V_s) &\ge \sqrt{\max(0,\min(\sigma_{\min}(V_{s-1})^2, 1-\poly(d)\varepsilon_{\machine}) - \sigma_{\max}(V_{s-1})\sqrt{\varepsilon_{\pca}})}.
\end{align*}
\label{lma:condition-number-bound}
\end{lemma}
\begin{proof}
    Let $Q$ be an orthonormal basis for the column space of $V_{s-1}$ and let $V_{s-1} = QR$ for a matrix $R$ with $\sigma_{\max}(R) = \sigma_{\max}(V_{s-1}) = \alpha_{s-1}$ and $\sigma_{\min}(R) = \sigma_{\min}(V_{s-1}) = \beta_{s-1}$. We have that $\opnorm{\T{Q}v_s'}^2 = \opnorm{Q\T{Q}v_s'}^2 = \opnorm{\Proj_{\colspace(V_{s-1})} v_s'}^2 \le \varepsilon_{\pca}$. Let $x \in \R^{s}$ be an arbitrary unit vector. Let $x_1 \in \R^{s-1}$ and $x_2 \in \R$ be such that $x = (x_1, x_2)$. Now,
    \begin{align*}
        \opnorm{V_s x}^2 &= \opnorm{QRx_1 + v'_s x_2}^2 = \opnorm{Rx_1}^2 + x_2^2\opnorm{v'_s}^2 + (2x_2) \T{x_1}\T{R}\T{Q}v'_s\\
        &\le \alpha_{s-1}^2\opnorm{x_1}^2 + (1+\poly(d)\varepsilon_{\machine})x_2^2 + (2|x_2|)\alpha_{s-1}\opnorm{x_1}\sqrt{\varepsilon_{\pca}}\\
        &\le \max(\alpha^2_{s-1}, 1+\poly(d)\varepsilon_{\machine}) + \alpha_{s-1}\sqrt{\varepsilon_{\pca}}(\opnorm{x_1}^2 + |x_2|^2)
    \end{align*}
    which implies that $\opnorm{V_s} \le \max(\alpha_{s-1}, 1+\poly(d)\varepsilon_{\machine}) + \sqrt{\varepsilon_{\pca}}$. Similarly,
    \begin{align*}
        \opnorm{V_s x}^2 &= \opnorm{QRx_1 + v'_s x_2}^2 = \opnorm{Rx_1}^2 + x_2^2\opnorm{v'_s}^2 + (2x_2)\T{x_1} \T{R}\T{Q}v'_s\\
        &\ge \beta_{s-1}^2\opnorm{x_1}^2 + (1 - \poly(d)\varepsilon_{\machine})x_2^2 - 2|x_2|\opnorm{x_1}\alpha_{s-1}\sqrt{\varepsilon_{\pca}}\\
        &\ge \min(\beta_{s-1}^2, 1-\poly(d)\varepsilon_{\machine}) - \alpha_{s-1}\sqrt{\varepsilon_{\pca}}.
    \end{align*}
    Hence, $\sigma_{\min}(V_s) \ge  \sqrt{\max(0,\min(\sigma_{\min}(V_{s-1})^2, 1-\poly(d)\varepsilon_{\machine}) - \sigma_{\max}(V_{s-1})\sqrt{\varepsilon_{\pca}})}$.
\end{proof}
Conditioned on the event that $\opnorm{\Proj_{\colspace(V_{s-1})}v'_s}^2 \le \varepsilon_{\pca}$ and $\opnorm{v'_s}^2 = (1 \pm \poly(d)\varepsilon_{\machine})$ for all $s=1,\ldots,k$, from the above lemma, we obtain that $\opnorm{V_s} \le 1 + \poly(d)\varepsilon_{\machine} + k\sqrt{\varepsilon_{\pca}}$. If $\varepsilon_{\pca} \le 1/\poly(k)$ and $\varepsilon_{\machine} \le 1/\poly(d)$, then $\opnorm{V_s} \le 2$ for all $s=1,\ldots,k$ which in turn implies that for all $s=1,\ldots,k$,
$
    \sigma_{\min}(V_s) \ge \sqrt{1 - \poly(d)\varepsilon_{\machine} - 2k\sqrt{\varepsilon_{\pca}}} \ge 1/2
$
assuming $\varepsilon_{\pca} \le 1/\poly(k)$ and $\varepsilon_{\machine} \le 1/\poly(d)$. 

Hence, $\kappa(V_s) \le 4$ for all $s=1,\ldots,k$ in Algorithm~\ref{alg:modified-lazy-svd} conditioned on the success of all the calls to $\AppxPCA$. Thus, we can assume that given any vector $x$, we can compute a vector $y$ on a floating point computer with precision $\varepsilon_{\machine} \le 1/\poly(d)$ such that
$
    \opnorm{y - M_s x} \le O(\varepsilon_{\machine}\poly(d) + \varepsilon_M)\opnorm{M}\opnorm{x}.
$

Let $A \in \R^{n \times d}$ be an arbitrary matrix with $n \ge d$. Define $M = \T{A}A$ to be a $d\times d$ matrix. Given any vector $x$, we can compute a vector $y$ satisfying
$
    \opnorm{\T{A}A x - y} \le O(\varepsilon_{\machine}\poly(n)\opnorm{A}^2\opnorm{x}). 
$
As $\opnorm{A}^2 = \opnorm{M}$, we thus have that for any $x$, we can compute $y$ satisfying $O(\varepsilon_{\text{machine}}\poly(d)\opnorm{M}\opnorm{x})$. Thus, $\varepsilon_{M}$ as defined above can be taken as $\varepsilon_{\machine}\poly(d)$. Let $\kappa = \lambda_1(M)/\lambda_{k+1}(M)$. By definition of eigenvalues, for any matrix $V$ with at most $k$ columns, we have
$
    \opnorm{(I-\Proj_{\colspace(V)})M(I - \Proj_{\colspace(V))})} \ge \lambda_{k+1}. 
$
Hence for all $s=1,\ldots,k$, $\opnorm{M_s} \ge \opnorm{M}/\kappa$. Thus, given any vector $x$, we can compute a vector $y$ satisfying
$
    \opnorm{y - M_s x} \le O(\varepsilon_{\machine}\poly(n)\kappa)\opnorm{M_s}\opnorm{x}.
$

Finally, we have the main theorem showing the stability of the LazySVD algorithm.

\begin{proof}[Proof of Theorem~\ref{thm:lazysvd-stability}]
Note that the algorithm in Lemma~\ref{lma:apxpca} satisfies the $\AppxPCA_{\varepsilon, \varepsilon_{\pca}, \eta}$ definition. Thus, if $\varepsilon_{\pca} \le \poly(\varepsilon/d\kappa)$, by Theorem~4.1 of \cite{allen2016lazysvd}, Algorithm~\ref{alg:modified-lazy-svd} when run on the $d \times d$ matrix $\T{A}A$ outputs a matrix $V_k$ such that with probability $\ge 1 - \eta$,
$\opnorm{(I-\Proj_{\colspace(V_k)})\T{A}A(I-\Proj_{\colspace(V_k)})} \le \frac{1}{1-\varepsilon}\sigma_{k+1}(A)^2
$
and for all $p' \ge 1$,
\[
    \|{(I-\Proj_{\colspace(V_k)})\T{A}A(I-\Proj_{\colspace(V_k)})}\|_{S_{p'}} \le (1+O(\varepsilon))(\sum_{i=k+1}^d \sigma_i(A)^{2p'})^{1/p'}.
\]
Thus, we have $\opnorm{A(I-\Proj_{\colspace(V_k)})} \le (1+O(\varepsilon))\sigma_{k+1}(A)$ and 
using the fact that $\|\T{A}A\|_{S_p}^p = \|A\|_{S_{2p}}^{2p}$, for all $p \ge 2$,
$
    \|A(I-\Proj_{\colspace(V_k)})\|_{S_p} \le (1+ O(\varepsilon))\|A - A_k\|_p.
$
We additionally have $\kappa(V_k) \le 4$ from Lemma~\ref{lma:condition-number-bound}.

\subparagraph*{Runtime Analysis.} In each iteration of Algorithm~\ref{alg:modified-lazy-svd}, we require $O(\varepsilon^{-1/2}\poly(\log(d\kappa/\eta\varepsilon)))$ matrix vector products with the matrices $A$ and $\T{A}$. For iterations $s=1,\ldots,k$, we solve $O(\varepsilon^{-1/2}\poly(\log(d\kappa/\eta\varepsilon)))$ least squares problems on a fixed $d \times s$ matrix and different label vectors. Thus, the overall time complexity of the algorithm is
\begin{align*}
    O\left(\frac{\nnz(A)k}{\sqrt{\varepsilon}}\poly(\log (d\kappa/\eta\varepsilon)) + d\poly(k, \log(d\kappa/\eta\varepsilon))\right). &\qedhere
\end{align*}
\end{proof}
\section*{Acknowledgments} We would like to thank Cameron Musco, Christopher Musco, and Aleksandros Sobczyk for helpful discussions. We thank a Simons Investigator Award and NSF CCF-2335411 for partial support.

\bibliography{main}
\appendix
\onecolumn
\section{Time Complexity of SVD in the Real RAM model}\label{sec:svd-complexity}
Consider an $m \times n$ matrix $A$ where $m \le n$. We can compute the matrix $M = \T{A}A$ in time $O(nm^{\omega-1})$, where $\omega$ is the matrix multiplication exponent. Using the eigen decomposition algorithm of \cite{banks2022pseudospectral}, we can then compute a matrix $V$ and a diagonal matrix $D$ satisfying $\opnorm{M - VDV^{-1}} \le \opnorm{M}/\poly(n)$ in time $\tilde{O}(m^{\omega})$. Although the matrix $M$ is symmetric, the matrix $V$ output by the algorithm may not be orthonormal. In the real RAM model, we can perform the following changes to their algorithm:
\begin{enumerate}
    \item The Ginibre perturbation step is replaced with the symmetric Gaussian Orthogonal Ensemble perturbation as mentioned in Remark~6.1 of \cite{banks2022pseudospectral}.
    \item In step 5 of the algorithm \textsf{EIG} in \cite{banks2022pseudospectral}, after computing the orthonormal matrices $\tilde{Q}_{+}$ and $\tilde{Q}_{-}$, we modify $\tilde{Q}_{-}$ to an orthonormal basis of the column space of the matrix $(I-\tilde{Q}_{+}\T{\tilde{Q}_{+}})\tilde{Q}_{-}$. This ensures that $\T{\tilde{Q}_{+}}{\tilde{Q}_{-}} = 0$, while preserving the properties of $\tilde{Q}_{-}$ guaranteed by the algorithm \textsf{DEFLATE}. Note that the matrix $\tilde{Q}_{-}$ can be updated in time $\tilde{O}(n^{\omega})$ in the real RAM model.
\end{enumerate}
Thus, the algorithm of \cite{banks2022pseudospectral} can be used to compute an orthonormal matrix $V$ and a diagonal matrix $D$ such that $\opnorm{M - VD\T{V}} \le \opnorm{M}/\poly(n)$ in $\tilde{O}(nm^{\omega-1})$ time in the real RAM model.

If we define $U = AV \cdot D^{-1/2}$, then $U, D^{1/2}, \T{V}$ is an approximate singular value decomposition of the matrix $A$, where the matrices $U, V$ are orthonormal up to a $1/\poly(n)$ error. Since the matrix $AV$ can be computed in $\tilde{O}(nm^{\omega-1})$, we obtain that SVD of a well conditioned matrix can be computed in $\tilde{O}(nm^{\omega-1})$ time.
\section{Missing Proofs from Section~\ref{sec:main-algorithm}}
We use the following fact, called the pinching inequality, about $\lp{ \cdot }$ norms frequently throughout the proof:
\begin{theorem}[Pinching Inequality \cite{bhatia2002pinchings,bakshi2022low}]
    Let $A$ be an arbitrary matrix and let $P$ and $Q$ be arbitrary orthogonal projection matrices of appropriate dimensions. Then
    \begin{align*}
        \lp{PAQ}^p + \lp{(I-P)A(I-Q)}^p \le \lp{A}^p.
    \end{align*}
    \label{thm:pinching-inequality}
\end{theorem}
We also use the following simple lemma in our proof.
\begin{lemma}
    Let $a, b, c \ge 0$ with $b \ge c$ be such that $a^2 \ge b^2 - \gamma c^2$ for a constant $\gamma < 1/2$. Then for any $p \ge 1$, $a^p \ge b^p - \gamma pc^2 b^{p-2}$.
    \label{lma:inequality-2-to-p}
\end{lemma}
\begin{proof}
    From $a^2 \ge b^2 - \gamma c^2$, we get $a^p \ge b^p(1 - \gamma c^2/b^2)^{p/2}$. As $\gamma c^2/b^2 \le 1/2$, we have $1 - \gamma c^2/b^2 \ge \exp(-2\gamma c^2/b^2)$ which implies
    \begin{align*}
        a^p \ge b^p \exp(-2\gamma c^2/b^2)^{p/2} = b^p\exp(-\gamma pc^2/b^2) \ge b^p(1 - \gamma pc^2/b^2) = b^p - \gamma p c^2 b^{p-2}. &\qedhere
    \end{align*}
\end{proof}
To prove that the Algorithm~\ref{alg:schatten-norm-subspace-apx} outputs a $1+\varepsilon$ approximate Schatten-$p$ norm low rank approximation, we proceed by case analysis: In the first case, we assume that the matrix $A$ does not have any large singular values and in the second case, we assume that the matrix has large singular values.
\subsection{No Large Singular Values}
Assume that $\sigma_1(A) \le (1 + 1/p)\sigma_{k+1}$. Algorithm~\ref{alg:schatten-norm-subspace-apx} runs the Block Krylov Iteration Algorithm two times with the following parameters:
\begin{enumerate}
	\item Block Size $b = k$ for $O(q\polylog n)$ iterations. Let $T_1$ be the time complexity of the algorithm and $Z_1$ be the output.
	\item Block size $b = b' + k$ for $O(\sqrt{p}\log(n/\varepsilon))$ iterations where $b' = \ceil{(3/2)\max(1,k/q^2\varepsilon)}$. Let $T_2$ be the time complexity of the Block Krylov algorithm with these parameters and $Z_2$ be the output.
\end{enumerate}
From Theorem~\ref{thm:block-krylov}, with a large probability, for all $i=1,\ldots,k$,
\begin{align*}
	\sigma_i(\T{A}Z_1)^2 \ge \opnorm{\T{A}(Z_1)_{*i}}^2 \ge \sigma_i^2 - \frac{c}{q^2}\sigma_{k+1}^2
\end{align*}
for a small enough constant $c$ (the value of $c$ depends on the constant factor in the iteration count). By Lemma~\ref{lma:inequality-2-to-p}, the above inequality implies that
$
	\sigma_i^p(\T{A}Z_1) \ge \sigma_i^p - \frac{cp}{q^2}\sigma_i^{p-2}\sigma_{k+1}^2
$
and therefore $\lp{\T{A}Z_1}^p = \sum_{i=1}^k \sigma_i(\T{A}Z_1)^p \ge \lp{A_k}^p - (cp/q^2)\sum_{i=1}^k\sigma_i^{p-2}\sigma_{k+1}^2$. Using the assumption that $\sigma_i \le (1+1/p)\sigma_{k+1}$ for all $i = 1,\ldots,k$, we obtain that
$
	\lp{\T{A}Z_1}^p \ge \lp{A_k}^p - (ecpk/q^2)\sigma_{k+1}^p
$
as $(1 + 1/p)^{p-2} \le e$ for all $p \ge 1$. We pick $c$ such that $ec \le 1$ which implies $\lp{\T{A}Z_1}^p \ge \lp{A_k}^p - (pk/q^2)\sigma_{k+1}^p$. 

If $(k/q^2)\sigma_{k+1}^p \le \varepsilon\lp{A - A_k}^p$, then
$
	\lp{A}^p - \lp{\T{A}Z_1\T{Z_1}}^p \le (1+p\varepsilon)\lp{A - A_k}^p.
$
By the pinching inequality (Theorem~\ref{thm:pinching-inequality}), and noting that $Z_1\T{Z_1}AW_1\T{W_1} = Z_1\T{Z_1}A$, we then have
\begin{align*}
    \lp{A(I-W_1\T{W_1})}^p \le \lp{A}^p - \lp{Z_1\T{Z_1}A}^p \le (1+p\varepsilon)\lp{A-A_k}^p
\end{align*}
which implies $\lp{A(I-W_1\T{W_1})} \le (1+\varepsilon)\lp{A - A_k}$. Suppose $(k/q^2)\sigma_{k+1}^p > \varepsilon\lp{A-A_k}^p$. Note that this implies $k/q^2 > \varepsilon$. We then have
\begin{align*}
	\sigma_{k+1}^p \ge \frac{\varepsilon q^2}{k}\lp{A - A_k}^p \ge \sigma_{b' + k}^p b' \frac{\varepsilon q^2}{k}
\end{align*}
and therefore that $\sigma_{k+1} \ge \sigma_{b' + k}(\varepsilon q^2b'/k)^{1/p}$. Note that $b' = \ceil{(3/2)\max(1, k/q^2\varepsilon)}$ is picked such that 
$
	b' \varepsilon q^2/k \ge 3/2.
$
Therefore, $\sigma_{k+1} \ge (3/2)^{1/p}\sigma_{b'+ k}$. Now $\gap = \frac{\sigma_k}{\sigma_{b'+k+1}} - 1 \ge \frac{\sigma_{k+1}}{\sigma_{b'+k}}-1 \ge (3/2)^{1/p} - 1 \ge 1/(5p)$.  By the gap-dependent convergence guarantee in Theorem~\ref{thm:block-krylov}, the matrix $Z_2$ satisfies
$
	\opnorm{\T{A}(Z_2)_{*i}}^2 \ge \sigma_i^2 - \poly(\varepsilon/n)\sigma_{k+1}^2
$
for all $i = 1,\ldots,k$. By picking the degree of $\poly(\varepsilon/n)$ to be appropriately large, we obtain 
$\lp{A}^p - \lp{\T{A}Z_2\T{Z}_2}^p \le (1+p\varepsilon)\lp{A - A_k}^p
$
which then using the pinching inequality implies that $\|A(I-W_2\T{W_2})\|_p^p \le (1+p\varepsilon)\|A - A_k\|_p^p$. Now,
$
	T_1 = \tilde{O}(q T(n, k))
$
and
$
	T_2 = \tilde{O}(\sqrt{p} T(n, b' + k))
$
and the time complexity of Algorithm~\ref{alg:schatten-norm-subspace-apx} is $O(T_1+T_2)$. We now show that with $q$ and $b'$ picked in the algorithm, the time complexity is as stated in Theorem~\ref{thm:final-theorem-our-algorithm}.

\subsubsection{\texorpdfstring{$k \le \varepsilon n^{\alpha}$}{k less than eps time n-power-alpha}}
If $k \le \varepsilon n^{\alpha}$, we can set $q = \sqrt{p}$ and $b' = (3/2)n^{\alpha}$ to obtain that $b' \varepsilon q^2/k \ge 3/2$. We then have $T_1 = \tilde{O}(\sqrt{p}n^2)$ and $T_2 = \tilde{O}(\sqrt{p}n^2)$ and therefore the time complexity of the algorithm is $\tilde{O}(\sqrt{p}n^2)$.

\subsubsection{\texorpdfstring{$\varepsilon n^{\alpha} \le k \le n^{\alpha}$}{k < n-power-alpha}}
Setting $q = \max(\sqrt{p}, \sqrt{\frac{3k}{2\varepsilon b'}})$ suffices. Let $b' = n^{\theta}$ for some $\theta \ge \alpha$. We then have $T(n, k + b') = O(n^{2-\alpha\beta}n^{\theta \beta})$ and $T_2 = \tilde{O}(\sqrt{p} n^{2 - \alpha\beta}n^{\theta \beta})$. Set $\theta$ such that
$
    n^{\theta} = \max(n^{\alpha}, ({n^{2\alpha\beta}k}/{p\varepsilon})^{\frac{1}{1+2\beta}}).
$
Then $q \le O(\max(\sqrt{p}, p^{\frac{1}{2(1+2\beta)}}(k/n^{\alpha}\varepsilon)^{\frac{\beta}{1+2\beta}}))$ and $T_1 = \tilde{O}(\sqrt{p}n^2 + p^{\frac{1}{2(1+2\beta)}}(k/n^{\alpha}\varepsilon)^{\frac{\beta}{1+2\beta}} n^2)$ and \[T_2 = \tilde{O}(\sqrt{p}n^{2-\alpha\beta}n^{\theta\beta}) = \tilde{O}(\sqrt{p}n^2 + p^{\frac{1}{2(1+2\beta)}}n^2 ({k}/{(n^{\alpha}\varepsilon)})^{\frac{\beta}{1+2\beta}})\] showing that the overall complexity of the algorithm is
$
    \tilde{O}(\max(\sqrt{p}n^2, p^{\frac{1}{2(1+2\beta)}}n^2 \left({k}/{(n^{\alpha}\varepsilon)}\right)^{\frac{\beta}{1+2\beta}})).
$
\subsubsection{\texorpdfstring{$k \ge n^{\alpha}$}{k > n-power-alpha}}
For any $r \ge n^{\alpha}$, $T(n, r) = n^{2 - \alpha\beta} r^{\beta}$. Note $b' = (3/2)\max(1, k/q^2\varepsilon)$ and that
\begin{align*}
	T(n,k) = n^{2-\alpha\beta}k^{\beta}\, \text{and}\, T(n, b' + k) = n^{2 - \alpha\beta}(b'+k)^{\beta}.
\end{align*}
Recall $\beta = \frac{\omega-2}{1-\alpha} > 0$ and we have $q = p^{1/2(1+2\beta)}/\varepsilon^{\beta/(1+2\beta)}$ in Algorithm~\ref{alg:schatten-norm-subspace-apx}. Note that $q \le 1/\sqrt{\varepsilon}$ for all $p \le 1/\varepsilon$. Hence, for all $p \le 1/\varepsilon$, $q \le 1/\sqrt{\varepsilon}$ and $k/q^2\varepsilon \ge k$ which implies that $b' = (3/2)k/q^2\varepsilon \ge k$. Now,
\begin{align*}
	T_1 = \tilde{O}\left(\frac{p^{\frac{1}{2(1+2\beta)}}}{\varepsilon^{\frac{\beta}{1+2\beta}}}n^{2 - \alpha \beta}k^{\beta}\right)
\quad 
\text{and}
\quad
	T_2 = \tilde{O}\left(\sqrt{p}n^{2-\alpha\beta}(b'+k)^{\beta}\right) = \tilde{O}\left(\frac{p^{\frac{1}{2(1+2\beta)}}}{\varepsilon^{\frac{\beta}{1+2\beta}}} n^{2-\alpha\beta}k^{\beta}\right).
\end{align*}
Thus, $T_1+T_2 = \tilde{O}((p^{1/2}\varepsilon^{-\beta})^{1/(1+2\beta)} n^{2-\alpha\beta}k^{\beta})$.
\subsection{Proof with Large Singular Values}
In this section, we prove Theorem~\ref{thm:final-theorem-our-algorithm} when the matrix $A$ has large singular values.

Assume that $\sigma_1 \ge (1+1/5p)\sigma_{k+1}$ (Note that the case analysis we do is not disjoint!). Let $\ell$ be the largest integer such that $\sigma_{\ell} \ge (1+1/5p)\sigma_{k+1}$ and $\sigma_{\ell} \ge (1+\varepsilon/n)\sigma_{\ell+1}$. If such an $\ell$ does not exist, then $\sigma_1 \le (1+\varepsilon/n)^k(1+1/5p)\sigma_{k+1} \le (1 + 2k\varepsilon/n + 1/5p)\sigma_{k+1}$. Assuming $p \le \log n/\varepsilon$, we have $2k\varepsilon/n \le 1/2p$ as long as $k \le n/4\log n$ which then implies $\sigma_1 \le (1 + 1/p)\sigma_{k+1}$. The correctness of the algorithm follows from the previous section.  Thus we assume that there exists an $\ell \in [k]$ defined as above. We have by definition of $\ell$ that
$
	\sigma_{\ell+1} \le (1+\varepsilon/n)^{k}(1+1/5p)\sigma_{k+1} \le (1 + 1/5p + 2k\varepsilon/n)\sigma_{k+1} \le (1+1/p)\sigma_{k+1}.
$

Now define $\gap_{\ell} = \frac{\sigma_{\ell}}{\sigma_{k+1}} - 1$ and by definition of $\ell$, we have $\gap_{\ell} \ge 1/5p$. As $q \ge \sqrt{p}$ in Algorithm~\ref{alg:schatten-norm-subspace-apx}, by Theorem~\ref{thm:block-krylov}, the matrix $Z_1$ satisfies
$
	\sigma_i(\T{Z_1}A)^2 \ge \opnorm{\T{A}(Z_1)_{*i}}^2 \ge \sigma_i^2 - \poly(\varepsilon/n)\sigma_{\ell+1}^2 \ge \sigma_i^2 - (1+1/p)^2{\poly(\varepsilon/n)\sigma_{k+1}^2}
$
for all $i \in [\ell]$. For $i$ such that $\ell+1 \le i \le k$, we have
\begin{align}
	\sigma_i(\T{Z_1}A)^2 \ge \opnorm{\T{A}(Z_1)_{*i}}^2 \ge \sigma_i^2 - \frac{1}{q^2}\sigma_{k+1}^2.
 \label{eqn:guarantee-for-small-svs}
\end{align}
Let $W_1$ be an orthonormal basis for the column space of $\T{A}Z_1$. We will now bound $\lp{A(I-W_1\T{W_1})}$. By the triangle inequality,
$
	\lp{A(I-W_1\T{W_1})} \le \lp{A_{\ell}(I-W_1\T{W_1})} + \lp{(A-A_{\ell})(I-W_1\T{W_1})}.
$
We now bound both the terms separately. As $\text{rank}(A_{\ell}(I-W_1\T{W_1})) \le \ell \le k$, we have $\lp{A_{\ell}(I-W_1\T{W_1})} \le \sqrt{k}\frnorm{A_{\ell}(I-W_1\T{W_1})}$ for any $p \in [1,\infty)$. We now use the following lemma which we prove in the appendix.
\begin{lemma}
	Let $Z$ be a rank-$k$ orthonormal matrix satisfying
$
		\opnorm{\T{A}(Z)_{*i}}^2 \ge \sigma_i^2 - \poly(\varepsilon/n)\sigma_{k+1}^2
$
	for all $i=1,\ldots, \ell$ for some $\ell \le k$. Let $W$ be an orthonormal basis for the column space of $\T{A}Z$. Then
	\begin{align*}
		\frnorm{A_{\ell}(I-W\T{W})}^2 \le \frac{\ell \poly(\varepsilon/n)\sigma_{k+1}^2}{1 - \sigma_{\ell+1}^2/\sigma_{\ell}^2}.
	\end{align*}
 \label{lma:error-on-top}
\end{lemma}
\begin{proof}
	Let $\widetilde{W}$ be a basis for the column space of $\T{A}Z_{*1:\ell}$. Clearly, $\text{span}(\widetilde{W}) \subseteq \text{span}(W)$ which then implies
	\begin{align*}
		\frnorm{A_{\ell}(I-W\T{W})} \le \frnorm{A_{\ell}(I-\widetilde{W}\T{\widetilde{W}})}.
	\end{align*}
	In the rest of the proof, we bound $\frnorm{A_{\ell}(I-\widetilde{W}\T{\widetilde{W}})}$. First, we have
\begin{align*}
	\frnorm{A_{\ell}\widetilde{W}\T{\widetilde{W}}}^2 = \frnorm{\Sigma_{\ell}\T{V_{\ell}}\widetilde{W}\T{\widetilde{W}}}^2 &= \frnorm{\Sigma_{\ell} (\T{V_{\ell}}\widetilde{W})}^2\\
	&= \sum_{i=1}^{\ell} \sigma_i^2 \opnorm{\T{(V_{\ell}}\widetilde{W})_{i*}}^2\\
	&= \sum_{i=1}^{\ell}\sigma_i^2 - \sum_{i=1}^{\ell}\sigma_i^2(1 - \opnorm{\T{(V_{\ell}}\widetilde{W})_{i*}}^2)\\
	&\le \frnorm{A_{\ell}}^2 - \sigma_{\ell}^2(\ell - \frnorm{\T{V_{\ell}}\widetilde{W}}^2)
	\end{align*}
	which then implies $\ell - \frnorm{\T{V_{\ell}}\widetilde{W}}^2 \le (\frnorm{A_{\ell}}^2 - \frnorm{A_{\ell}\widetilde{W}\T{\widetilde{W}}}^2)/\sigma_{\ell}^2$.
Now, we write $\frnorm{A\widetilde{W}\T{\widetilde{W}}}^2 = \frnorm{A_{\ell}\widetilde{W}\T{\widetilde{W}}}^2 + \frnorm{(A-A_{\ell})\widetilde{W}\T{\widetilde{W}}}^2$ using the Pythagorean theorem. We can further write $\frnorm{(A-A_{\ell})\widetilde{W}\T{\widetilde{W}}}^2 = \frnorm{\Sigma_{\setminus \ell}\T{V_{\setminus \ell}}\widetilde{W}}^2$. Since $\ell = \frnorm{\widetilde{W}}^2 = \frnorm{\T{V_{\ell}}\widetilde{W}}^2 + \frnorm{\T{V_{\setminus \ell}}\widetilde{W}}^2$, we get $\frnorm{\T{V_{\setminus \ell}}\widetilde{W}}^2 = (\ell - \frnorm{\T{V_{\ell}}\widetilde{W}}^2)$ and therefore that
\begin{align*}
	\frnorm{(A-A_{\ell})\widetilde{W}\T{\widetilde{W}}}^2 = \frnorm{\Sigma_{\setminus \ell}\T{V_{\setminus \ell}}\widetilde{W}}^2 \le \sigma_{\ell+1}^2(\ell - \frnorm{\T{V_{\ell}}\widetilde{W}}^2).
\end{align*}
Overall,
$
	\frnorm{A\widetilde{W}\T{\widetilde{W}}}^2 \le \frnorm{A_{\ell}\widetilde{W}\T{\widetilde{W}}}^2 + ({\sigma_{\ell+1}^2}/{\sigma_{\ell}^2})(\frnorm{A_{\ell}}^2 - \frnorm{A_{\ell}\widetilde{W}\T{\widetilde{W}}}^2).$
We now use the fact that $\frnorm{A\widetilde{W}}^2 \ge \frnorm{\T{A}Z_{*1:\ell}}^2 \ge \frnorm{A_{\ell}}^2 - \ell \poly(\varepsilon/n)\sigma_{k+1}^2$ and obtain,
\begin{align*}
	\frnorm{A_{\ell}}^2 - \ell \poly(\varepsilon/n)\sigma_{k+1}^2 \le \frnorm{A_{\ell}\widetilde{W}\T{\widetilde{W}}}^2 + \frac{\sigma_{\ell+1}^2}{\sigma_{\ell}^2}(\frnorm{A_{\ell}}^2 - \frnorm{A_{\ell}\widetilde{W}\T{\widetilde{W}}}^2).
\end{align*}
Rearranging the inequality and using the  Pythagorean theorem,
\begin{align*}
	\frnorm{A_{\ell}(I-W\T{W})}^2 \le \frnorm{A_{\ell}(I-\widetilde{W}\T{\widetilde{W}})}^2 \le \frac{\ell \poly(\varepsilon/n)\sigma_{k+1}^2}{1 - \sigma_{\ell+1}^2/\sigma_{\ell}^2}.&\qedhere
\end{align*}
\end{proof}
Since $\sigma_{\ell} \ge (1+\varepsilon/n)\sigma_{\ell+1}$, we get $1 - \sigma_{\ell+1}^2/\sigma_{\ell}^2 \ge 1 - 1/(1+\varepsilon/n)^2 \ge \varepsilon/n$. Now using the above lemma,
$
	\frnorm{A_{\ell}(I-W_1\T{W_1})}^2 \le {\ell \poly(\varepsilon/n)\sigma_{k+1}^2}/{(\varepsilon/n)} \le \poly(\varepsilon/n)\sigma_{k+1}^2.
$
Hence,
$
	\lp{A(I-W_1\T{W_1})} \le \poly(\varepsilon/n)\sigma_{k+1} + \lp{(A-A_{\ell})(I-W_1\T{W_1})}.
$
Now, to bound the second term, we again use the triangle inequality to obtain
$
	\lp{(A-A_{\ell})(I-W_1\T{W_1})} \le \lp{Z_1\T{Z_1}(A-A_{\ell})(I-W_1\T{W_1})} + \lp{(I-Z_1\T{Z_1})(A-A_{\ell})(I-W_1\T{W_1})}.
$
Now, $Z_1\T{Z_1}A(I-W_1\T{W_1}) = 0$ since the columns of $W_1$ are a basis for the rows of $\T{Z_1}A$. Hence, 
\begin{align*}
\lp{Z_1\T{Z_1}(A-A_{\ell})(I-W_1\T{W_1})} &= \lp{Z_1\T{Z_1}A_{\ell}(I-W_1\T{W_1})} \le \lp{A_{\ell}(I-W_1\T{W_1})}\\
&\le \sqrt{k}\frnorm{A_{\ell}(I-W_1\T{W_1})} \le \poly(\varepsilon/n)\sigma_{k+1}.
\end{align*}
Thus, it only remains to bound $\lp{(I-Z_1\T{Z_1})(A-A_{\ell})(I-W_1\T{W_1})}$. Using the pinching inequality,
$
	\lp{(I-Z_1\T{Z_1})(A-A_{\ell})(I-W_1\T{W_1})}^p \le \lp{A-A_{\ell}}^p - \lp{Z_1\T{Z_1}(A-A_{\ell}) W_1\T{W_1}}^p.
$
By definition of the Schatten-$p$ norm, $\lp{A-A_{\ell}}^p = \sum_{i=l+1}^k \sigma_i(A)^p + \sum_{i=k+1}^n \sigma_i(A)^p$ and $\lp{Z_1\T{Z_1}(A-A_{\ell})W_1\T{W_1}}^p = \sum_{i=1}^k \sigma_i(\T{Z_1}(A-A_{\ell})W_1)^p$ where we used the fact that $\T{Z_1}(A-A_{\ell})W_1$ is a $k \times k$ matrix. Hence,
\begin{align*}
	\lp{A-A_{\ell}}^p - \lp{Z_1\T{Z_1}(A-A_{\ell})W_1\T{W_1}}^p &= \lp{A-A_k}^p + \sum_{i=\ell+1}^k \sigma_i(A)^p - \sum_{i=1}^{k}\sigma_i(\T{Z_1}(A-A_{\ell})W_1)^p.
\end{align*}
By Weyl's inequality, $\sigma_{i}(\T{Z_1}(A-A_{\ell})W_1) \ge \sigma_{i+\ell}(\T{Z_1}AW_1) - \sigma_{\ell+1}(\T{Z_1}A_{\ell}W_1)$. Now, using the fact that the matrix $\T{Z_1}A_{\ell}W_1$ has rank at most $\ell$, we obtain $\sigma_i(\T{Z_1}(A-A_{\ell})W_1) \ge \sigma_{i+\ell}(\T{Z_1}AW_1)$. Thus,
\begin{align*}
	\sum_{i=\ell+1}^k \sigma_i(A)^p - \sum_{i=1}^k \sigma_i(\T{Z}(A-A_{\ell})W_1)^p &\le \sum_{i={\ell+1}}^k \sigma_i(A)^p - \sigma_{i}(\T{Z_1}AW_1)^p\\
 &=\sum_{i={\ell+1}}^k \sigma_i(A)^p - \sigma_{i}(\T{Z_1}A)^p.
\end{align*}
We then have from Lemma~\ref{lma:inequality-2-to-p} and \eqref{eqn:guarantee-for-small-svs} that
$
	\sigma_{i}^p(\T{Z_1}A) \ge \sigma_i^p - \frac{cp}{q^2}\sigma_i^{p-2}\sigma_{k+1}^2
$
and therefore, $\sum_{i=\ell+1}^k \sigma_i(A)^p - \sigma_i(\T{Z_1}A)^p \le (cp/q^2)\sigma_{k+1}^2\sum_{i=\ell+1}^k \sigma_i^{p-2} \le (ecp/q^2)k\sigma_{k+1}^p$ using the fact that $\sigma_{\ell+1} \le (1 + 1/p)\sigma_{k+1}$. The constant factors in the algorithm are set such that $ec \le 1$ and therefore, we obtain overall that
\begin{align*}
	\lp{A-A_{\ell}}^p - \lp{Z_1\T{Z_1}(A-A_{\ell})W_1\T{W_1}}^p \le \lp{A-A_k}^p + (kp/q^2)\sigma_{k+1}^p.
\end{align*}
Now, just as in the proof in the case of ``no small singular values'', if $(k/q^2)\sigma_{k+1}^p \le \varepsilon\lp{A-A_k}^p$, we get
\begin{align*}
\lp{A(I-W_1\T{W_1})} &\le \poly(\varepsilon/n)\sigma_{k+1} + \poly(\varepsilon/n)\sigma_{k+1} + ((1+p\varepsilon)\lp{A-A_k}^{p})^{1/p}\\
&\le (1 + 2\varepsilon)\lp{A-A_k}.
\end{align*}
Otherwise, we have $(k/q^2)\sigma_{k+1}^p \ge \varepsilon\lp{A-A_k}^p \ge \varepsilon b' (\sigma_{k+b'})^p$.  As long as $b'\varepsilon q^2/k \ge 3/2$, we have $\gap = \sigma_{k}/\sigma_{k+b'+1} - 1 \ge 1/5p$. Again, as the Block Krylov algorithm is run with block size $k+b'$ for at least $\Omega(\sqrt{p}\log(n/\varepsilon))$ iterations, we obtain from Theorem~\ref{thm:block-krylov} that for all $i \in [k]$,
$
    \sigma_i(\T{Z_2}A)^2 \ge \opnorm{\T{A}(Z_2)_{*i}}^2 \ge \sigma_i^2 - \poly(\varepsilon/n)\sigma_{k+1}^2.
$
Using a proof similar to above, we get $\lp{A(I-W_2\T{W_2})} \le \poly(\varepsilon/n)\sigma_{k+1} + \lp{(I-Z_2\T{Z_2})(A-A_{\ell})(I-W_2\T{W_2})}$. Again, we have
\begin{align*}
    \lp{(I-Z_2\T{Z_2})(A-A_{\ell})(I-W_2\T{W_2})}^p &\le \lp{A-A_{\ell}}^p - \lp{Z_2\T{Z_2}(A-A_{\ell})W_2\T{W_2}}^p\\
    &\le \lp{A-A_k}^p + \sum_{i=\ell+1}^k \sigma_i(A)^p -\sigma_i(\T{Z_2}A)^p.
\end{align*}
As we have $\sigma_i(\T{Z_2}A)^2 \ge \sigma_i^2 - \poly(\varepsilon/n)\sigma_{k+1}^2$, we get $\sigma_i(\T{Z_2}A)^p \ge \sigma_{i}^p - p \cdot \poly(\varepsilon/n)\sigma_{i}^{p-2}\sigma_{k+1}^2$ from Lemma~\ref{lma:inequality-2-to-p}. Using the fact that for $i \ge \ell+1$, $\sigma_i \le (1+1/p)\sigma_{k+1}$, we get $\sum_{i=\ell+1}^k \sigma_i(A)^p - \sigma_i(\T{Z_2}A)^p \le \poly(\varepsilon/n)\sigma_{k+1}^p$. Thus,
$
     \lp{(I-Z_2\T{Z_2})(A-A_{\ell})(I-W_2\T{W_2})} \le (1+\poly(\varepsilon/n))^{1/p}\lp{A-A_k}
$
which overall implies that
$
    \lp{A(I-W_2\T{W_2})} \le (1 + \poly(\varepsilon/n))\lp{A-A_k}.
$
\subsection{Choosing between \texorpdfstring{$W_1$}{W1} and \texorpdfstring{$W_2$}{W2}}\label{subsec:picking-w1-w2}
Algorithm~\ref{alg:schatten-norm-subspace-apx} computes two rank-$k$ orthonormal matrices $W_1$ and $W_2$. Above, we proved that with high probability, one of the two matrices $W_1$ and $W_2$ is a  $(1+O(\varepsilon))$-approximation to the Schatten-$p$ low rank approximation problem. Computing the Schatten norms of $A(I-W_1\T{W_1})$ and $A(I-W_2\T{W_2})$ up to the desired accuracy is too slow. Thus, we need an alternate way to decide which matrix to output as the solution.

Note the Block Krylov iteration algorithm with block size $b'+k$ and the rank parameter $b'+k$ for $O(\sqrt{p}\log(n))$ iterations outputs estimates $\hat{\sigma}_{k}$ and $\hat{\sigma}_{b'+k}$ satisfying with probability $\ge 99/100$,
$
    \hat{\sigma}_i = (1 \pm 1/10p){\sigma}_i
$
for $i \in \set{k,b'+k}$. We do not have to spend more computation as we already compute the SVD of the $n \times O(\sqrt{p}(b'+k)\log(n/\varepsilon))$ matrix in the Block Krylov iteration algorithm when run with block size $b'+k$. 

Suppose $\hat{\sigma}_k \ge (1+1/2p)\hat{\sigma}_{b'+k}$. Conditioned on the above event about the accuracy of $\hat{\sigma}_i$, we get
\begin{align*}
    \sigma_k \ge \frac{\hat{\sigma}_k}{1+1/10p} \ge \frac{1+1/2p}{1+1/10p}\hat{\sigma}_{b'+k} \ge \frac{(1+1/2p)(1-1/10p)}{1+1/10p}\sigma_{b'+k} \ge (1+1/5p)\sigma_{b'+k}.
\end{align*}
Thus, $\gap = \sigma_{k}/\sigma_{b'+k+1} - 1\ge 1/5p$ which implies that $W_2$ is a $(1+\varepsilon)$-approximation from the analysis in previous sections. If $\hat{\sigma}_{k} < (1+1/2p)\hat{\sigma}_{b'+k}$, then
\begin{align*}
    \sigma_{b'+k} \ge \frac{\hat{\sigma}_{b'+k}}{1 + 1/10p} \ge \frac{\hat{\sigma}_k}{(1+1/10p)(1+1/2p)} \ge \frac{(1-1/10p)}{(1+1/10p)(1+1/2p)}\sigma_k \ge (1 - 7/10p)\sigma_k.
\end{align*}
which implies $\sigma_{b'+k} \ge (1-7/10p)\sigma_k$. We now have
$
    \lp{A-A_k}^p \ge b'\sigma_{k+b'}^p \ge b'(1-7/10p)^p\sigma_{k}^p \ge b'(e^{-7/5p})^p\sigma_k^p \ge b'e^{-7/5}\sigma_k^p. 
$
Hence, $\sigma_{k+1}^p \le e^{7/5}\lp{A-A_k}^p/b' \le (2\varepsilon q^2/3k) e^{7/5}\lp{A-A_k}^p$ using $b'\varepsilon q^2/k \ge 3/2$. 

In the case of ``no large singular values'', we have $\lp{A(I-W_1\T{W_1})}^p \le (1+3p\varepsilon)\lp{A-A_k}^p$ which then implies $\lp{A(I-W_1\T{W_1})} \le (1+O(\varepsilon))\lp{A-A_k}$. In the case of ``large singular values'', we have $\lp{A(I-W_1\T{W_1})} \le \poly(\varepsilon/n)\sigma_{k+1} + ((1+3p\varepsilon)\lp{A-A_k}^p)^{1/p} \le (1+4\varepsilon)\lp{A-A_k}$. 

Thus, an algorithm which outputs $W_2$ when $\hat{\sigma}_k \ge (1+1/2p)\hat{\sigma}_{b'+k}$ and $W_1$ otherwise is a $(1 + O(\varepsilon))$-approximation algorithm to the Schatten-$p$  norm low rank approximation.
\section{An Experiment}
We consider multiplying an $n \times n$ matrix with an $n \times d$ matrix while varying $d$. We set $n = 10{,}000$ and vary $d$ to take values in the interval $[10, 100]$. If $t_d$ is the median amount of time (over 5 repetitions) to compute the product of an $n \times n$ matrix with an $n \times d$ matrix, we obtain a color map~(Figure~\ref{fig:color-map}) of the values $(j/i)/(t_j/t_i)$ for $j \ge i$. If $(j/i)/(t_j/t_i)$ is large then $t_j$ is much smaller than $t_i (j/i)$ which is what we would expect if the matrix-multiplication time scales linearly with the dimension.
\begin{figure}
    \centering
    \includegraphics[width=\linewidth]{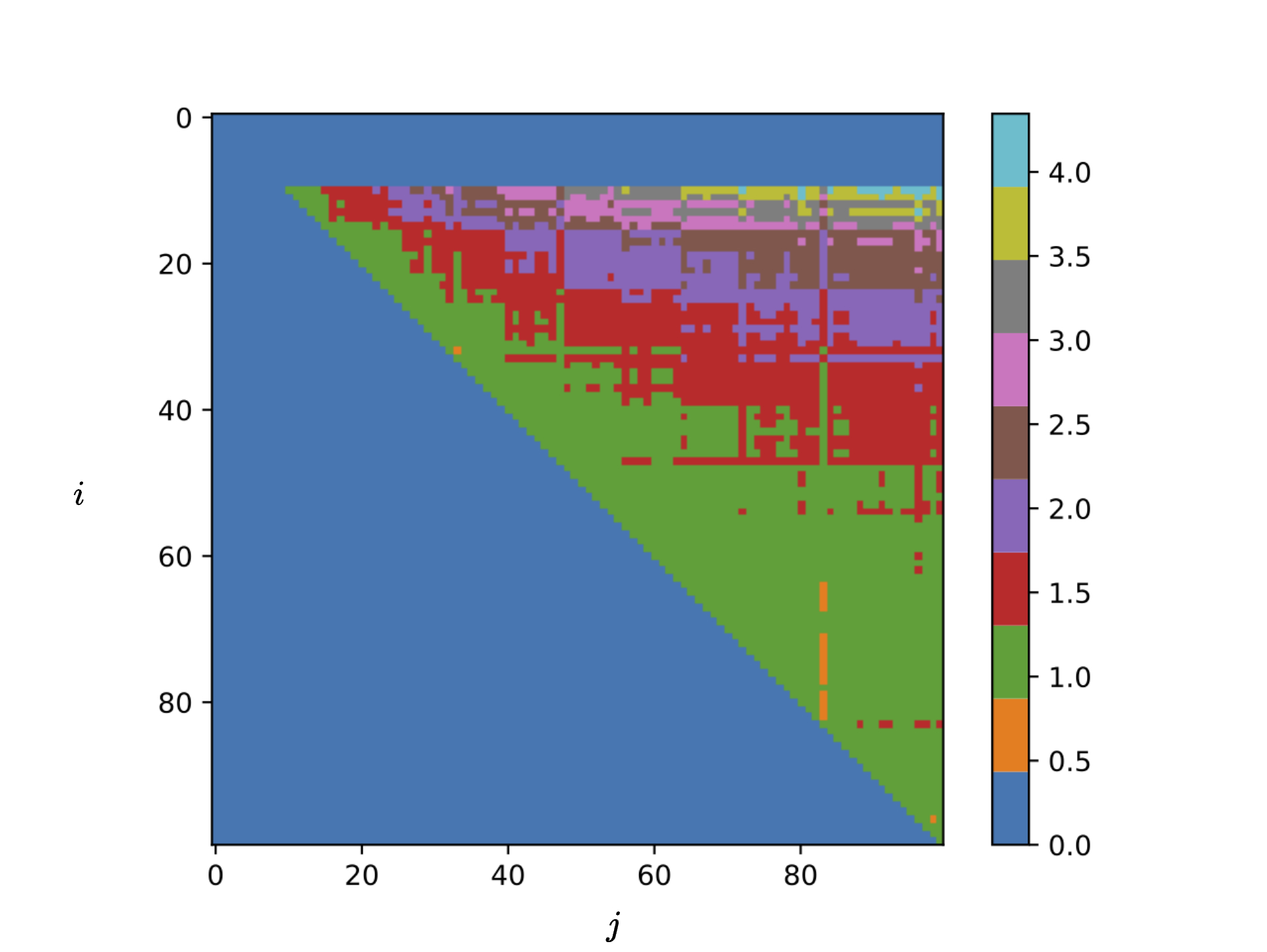}
    \caption{Color Map of $(j/i)/(t_j/t_i)$}
    \label{fig:color-map}
\end{figure}
The experiment was performed using NumPy on a machine with 2 cores. We see that fixing an $i$, as we increase $j$, $t_j$ becomes smaller compared to $t_i \cdot (j/i)$. Hence, it is advantageous to run with larger block sizes if it means that it reduces the number of iterations for which the smaller block size is to be run. In the proof of Theorem~\ref{thm:final-theorem-our-algorithm}, we see that if we increase the larger block to 4 times the original, then the number of iterations the smaller block size is to be run decreases to 0.5x the original. Based on the characteristics of the machine, we can obtain significant improvements over the parameters obtained by optimizing for matrix-vector products.
\end{document}